\documentclass{article}
\usepackage[utf8]{inputenc}
\usepackage[margin=1in]{geometry}

\RequirePackage{doi}
\usepackage[noadjust]{cite}
\usepackage{hyperref}
\usepackage{thm-restate}
\hypersetup{
    colorlinks=true,
    linkcolor=blue,
    filecolor=magenta,      
    urlcolor=cyan,
    pdftitle={Counting Small Induced Subgraphs with Hereditary Properties},
    pdfpagemode=FullScreen,
    }
\usepackage{blkarray}
\usepackage{bm}
\usepackage{rotating}
\usepackage{afterpage}
\usepackage{booktabs}
\usepackage{amsmath, amsthm}
\usepackage{amssymb}
\usepackage{enumitem}
\usepackage{graphicx}
\usepackage{mathdots}
\usepackage{color}
\usepackage{hyperref,url}
\usepackage{tikz}
\usetikzlibrary{matrix}
\usetikzlibrary{shapes}
\usetikzlibrary{arrows,decorations.markings, tikzmark}
\definecolor{darkgray}{rgb}{0.31,0.31,0.31}
    \definecolor[named]{lipicsGray}{rgb}{0.31,0.31,0.31}
    \definecolor[named]{lipicsBulletGray}{rgb}{0.60,0.60,0.60}
    \definecolor[named]{lipicsLineGray}{rgb}{0.50,0.50,0.50}
    \definecolor[named]{lipicsLightGray}{rgb}{0.85,0.85,0.85}
    \definecolor[named]{lipicsYellow}{rgb}{0.85,0.85,0.85}

\usetikzlibrary{calc,shapes,fit}

\tikzset{vertex/.style={circle, fill, inner sep=1.5pt, outer sep=1.9pt}}
\tikzset{cvertex/.style={circle, fill, inner sep=1pt, outer sep=1.3pt, lipicsGray!60}}
\tikzset{ccvertex/.style={circle, fill, inner sep=1.3pt, outer sep=1.7pt, lipicsGray!80}}
\tikzset{svertex/.style={circle,draw=white, line width=1.2pt,fill=black, inner sep=1.8pt, outer         sep=0pt}}

\tikzset{edge/.style={very thick}}
\tikzset{ccedge/.style={lipicsGray!50}}
\tikzset{cedge/.style={lipicsGray!75}}
\tikzset{dedge/.style={gray,thick, double=white, double distance=9pt}}
\tikzset{sdedge/.style={white,thick, double=black, double distance=1.2pt}}
\usepackage{mathtools}
\usepackage{stmaryrd}
\usepackage{tabularx}
\usepackage{xspace}

\newtheorem{theorem}{Theorem}
\newtheorem{obs}[theorem]{Observation}
\newtheorem{lemma}[theorem]{Lemma}
\newtheorem{fact}[theorem]{Fact}

\newtheorem{clm}{Claim}[theorem]
\theoremstyle{definition}
\newtheorem{defn}[theorem]{Definition}
\newtheorem{conjecture}[theorem]{Conjecture}

\newenvironment{proofclaim}[1][\unskip]{\noindent {\emph{Proof of Claim #1.\space}}}{\hfill$\triangleleft$\smallskip}

\newcommand{\id}{\mathsf{id}}

\newcommand{\W}[1]{\ensuremath{\mathsf{W[#1]}}}

\newcommand{\timesc}{\ensuremath{\times}}
\newcommand{\timescc}{\ensuremath{\times}}
\newcommand{\down}{\ensuremath{\!\!\downarrow}}

\newcommand{\homs}[2]{\mbox{\ensuremath{\mathsf{Hom}(#1 \to #2)}}}

\newcommand{\indsubs}[2]{\mbox{\ensuremath{\mathsf{IndSub}(#1 \to #2)}}}
\newcommand{\bipindsubs}[2]{\mbox{\ensuremath{\mathsf{BipIndSub}(#1 \to #2)}}}
\newcommand{\auts}[1]{\ensuremath{\mathsf{Aut}(#1)}}
\newcommand{\cphoms}[3]{\ensuremath{\mathsf{cp}\text{-}\mathsf{Hom}}(#1 \to_#2 #3)}

\newcommand{\cpindsubs}[3]{\ensuremath{\mathsf{cp}\text{-}\mathsf{IndSub}}(#1 \to_#2 #3)}
\newcommand{\cpbipindsubs}[3]{\ensuremath{\mathsf{cp}\text{-}\mathsf{BipIndSub}}(#1 \to_#2 #3)}

\newcommand{\clique}{\ensuremath{\textsc{Clique}}}

\newcommand{\cphomsprob}{\ensuremath{\textsc{cp-Hom}}}
\newcommand{\indsubsprob}{\ensuremath{\textsc{IndSub}}}

\newcommand{\bipindsubsprob}{\ensuremath{\textsc{BipIndSub}}}
\newcommand{\cpbipindsubsprob}{\ensuremath{\textsc{cp-BipIndSub}}}

\newcommand{\meagre}{meagre\xspace}
\newcommand{\implant}{implant\xspace}
\newcommand{\orb}{\ensuremath{\mathsf{Orb}}}
\newcommand{\stab}{\ensuremath{\mathsf{Stab}}}
\newcommand{\bG}{{\boldsymbol{\mathsf{G}}}}
\newcommand{\bH}{{\boldsymbol{\mathsf{H}}}}
\newcommand{\bB}{{\boldsymbol{\mathsf{B}}}}

\newcommand{\bI}{{\boldsymbol{\mathsf{I}}}}
\newcommand{\obstr}{\Pi}

\newcommand{\problem}[3] {\textbf{Input:} #1.\\
\textbf{Parameter:} #2.\\
\textbf{Output:} #3.}

\newcommand{\fptred}{\ensuremath{\leq^{\mathrm{fpt}}_{\mathrm{T}}}}
\newcommand{\tightred}{\ensuremath{\leq^{\mathrm{tfpt}}_{\mathrm{T}}}}

\title{Counting Small Induced Subgraphs with Hereditary Properties}
\author{Jacob Focke\thanks{Research supported by the European Research Council (ERC) consolidator grant No.~725978 SYSTEMATICGRAPH.} \\ CISPA \\ Helmholtz Center for Information Security   \\\texttt{jacob.focke@cispa.de} \and  Marc Roth\\
Department of Computer Science \\ University of Oxford \\ \texttt{marc.roth@cs.ox.ac.uk}}
\date{\today}

\begin{document}

\maketitle
\thispagestyle{empty}

\begin{abstract}
    We study the computational complexity of the problem $\#\textsc{IndSub}(\Phi)$ of counting $k$-vertex induced subgraphs of a graph $G$ that satisfy a graph property $\Phi$. 
    Our main result establishes an exhaustive and explicit classification for all hereditary properties, including tight conditional lower bounds under the Exponential Time Hypothesis (ETH):
    \begin{itemize}
        \item If a hereditary property $\Phi$ is true for all graphs, or if it is true only for finitely many graphs, then $\#\textsc{IndSub}(\Phi)$ is solvable in polynomial time.
        \item Otherwise, $\#\textsc{IndSub}(\Phi)$ is $\#\mathsf{W[1]}$-complete when parameterised by $k$, and, assuming ETH, it cannot be solved in time $f(k)\cdot |G|^{o(k)}$ for any function $f$.
    \end{itemize}
    This classification features a wide range of properties for which the corresponding detection problem (as classified by Khot and Raman [TCS 02]) is tractable but counting is hard. Moreover, even for properties which are already intractable in their decision version, our results yield significantly stronger lower bounds for the counting problem.
   
    As additional result, we also present an exhaustive and explicit parameterised complexity classification for all properties that are invariant under homomorphic equivalence.

    By covering one of the most natural and general notions of closure, namely, closure under vertex-deletion (hereditary), we generalise some of the earlier results on this problem. For instance, our results fully subsume and strengthen the existing classification of $\#\textsc{IndSub}(\Phi)$ for monotone (subgraph-closed) properties due to Roth, Schmitt, and Wellnitz [FOCS 20].
\end{abstract}

\clearpage

\tableofcontents
\thispagestyle{empty}
\newpage

\setcounter{page}{1}
\section{Introduction}\label{sec:intro}
Detection and counting of patterns in networks belong to the most well-studied problems in theoretical computer science and have applications in a diverse set of disciplines such as database theory~\cite{GroheSS01}, statistical physics~\cite{Kasteleyn61,TemperleyF61, Kasteleyn63}, and computational biology~\cite{SchreiberS05,GrochovK07}.
In this work, we focus on counting \emph{small} patterns in \emph{large} networks. Among others, this task is motivated by the computation of so-called significance profiles of network motifs which play a central role in the analysis of complex networks~\cite{Miloetal02, Miloetal04,Nogaetat08, Schilleretal15}. 

More formally, we consider the counting problem $\#\indsubsprob(\Phi)$ as introduced by Jerrum and Meeks~\cite{JerrumM15}.\footnote{In~\cite{JerrumM15}, the problem $\#\indsubsprob(\Phi)$ is called $\#\textsc{InducedUnlabelledSubgraphWithProperty}(\Phi)$.} Here, a graph property $\Phi$ is a function from the class of graphs to $\{0,1\}$ that is invariant under graph isomorphisms. If $\Phi(H)=1$ for a graph $H$, then $H$ \emph{satisfies} the property $\Phi$. For any fixed graph property~$\Phi$, the problem $\#\indsubsprob(\Phi)$ asks, on input a graph $G$ and a positive integer~$k$, to compute the number of $k$-vertex induced subgraphs $H$ in $G$ that satisfy $\Phi$. Observe that, for proper choices of $\Phi$, the problem $\#\indsubsprob(\Phi)$ encodes a variety of well-studied counting problems such as counting of $k$-cliques, $k$-independent sets, induced $k$-cycles, and, to name a more intricate example, $k$-graphlets, that is, connected $k$-vertex induced subgraphs. 

In recent years, the problem $\#\indsubsprob(\Phi)$ received significant attention~\cite{JerrumM15,JerrumM15density,Meeks16,JerrumM17,CurticapeanDM17,RothS18,DorflerRSW22,RothSW20}. All of the previous works had the common goal of classifying the \emph{parameterised} complexity of $\#\indsubsprob(\Phi)$ for a wide range of properties $\Phi$. More precisely, the task is to identify those properties $\Phi$ for which the problem becomes \emph{fixed-parameter tractable (FPT)}, i.e., solvable in time $f(k)\cdot |G|^{O(1)}$ for some computable function $f$. Note that a parameterised analysis of $\#\indsubsprob(\Phi)$ captures well the intuition that the size of the pattern~$k$ is significantly smaller than the size of the graph $G$, that is, we only aim for a running time which is polynomial in $|G|$ but may be super-polynomial in $k$.

Ideally, a complete classification of $\#\indsubsprob(\Phi)$ identifies not only the properties $\Phi$ for which the problem becomes FPT, but also establishes a hardness result for all remaining properties. A remarkable result due to Curticapean, Dell and Marx~\cite{CurticapeanDM17} shows that such a complete classification is possible: They prove that for every property $\Phi$, the problem $\#\indsubsprob(\Phi)$ is either fixed-parameter tractable or complete for the parameterised class $\#\W{1}$.\footnote{The class $\#\W{1}$ can be considered as a parameterised counting equivalent of $\mathrm{NP}$; a formal definition is provided in Section~\ref{sec:prelim}.} Unfortunately, their classification is implicit in the sense that, for most graph properties $\Phi$, it is not clear how to pinpoint the complexity of $\#\indsubsprob(\Phi)$. More precisely, even for simple and natural properties such as $\Phi(H)=1$ iff $H$ is bipartite, or $\Phi(H)=1$ iff $H$ is acyclic, the complexity of $\#\indsubsprob(\Phi)$ is not easily deducible from the aforementioned classification. The corresponding hardness proofs turned out to be a non-trivial task~\cite{RothS18}.
Subsequent work focused on finding \emph{explicit} criteria for tractability and hardness~\cite{RothS18,DorflerRSW22,RothSW20}. More details on the classification from~\cite{CurticapeanDM17} are given in Section~\ref{sec:techniques}.

The state of the art suggests that the only properties $\Phi$ for which $\#\indsubsprob(\Phi)$ is FPT are very restricted in the sense that they become ``eventually trivial''. More formally, we say that a property $\Phi$ is \emph{\meagre} if there exists a positive integer $B$, such that for each $k\geq B$ the property $\Phi$ is either constant false or constant true on the set of all $k$-vertex graphs. For example, the property $\Phi$ of having an even number of vertices is \meagre, and it is easy to see that $\#\indsubsprob(\Phi)$ is trivial to solve: On input $G$ and $k$, output $0$ if $k$ is odd, and output $\binom{|V(G)|}{k}$ if $k$ is even.
It is well-known that an analogue of the previous algorithm exists for every \meagre property; we make this formal in Section~\ref{sec:prelim}. 
Conversely, as stated in~\cite{RothSW20}, we conjecture that all non-\meagre properties yield hardness:
\begin{conjecture}\label{conj:main}
        Let $\Phi$ be a computable\footnote{We restrict ourselves to computable properties to avoid dealing with non-uniform fixed-parameter tractability.} graph property. If $\Phi$ is \meagre then $\#\indsubsprob(\Phi)$ is fixed-parameter tractable. Otherwise, $\#\indsubsprob(\Phi)$ is $\#\W{1}$-complete.
    \end{conjecture}
    
Despite significant effort, we are nowhere close to a resolution of Conjecture~\ref{conj:main} in its full generality. However, progress has been made for properties that satisfy certain closure criteria. For example, a result of Jerrum and Meeks~\cite{JerrumM15density} implies that  Conjecture~\ref{conj:main} is true for \emph{minor-closed} graph properties. After several partial results, Conjecture~\ref{conj:main} has recently also been established for the more general class of \emph{monotone} properties, that is, properties that are closed under the removal of vertices \emph{and} edges\footnote{To avoid confusion, we remark that in some literature (e.g.\ \cite{JerrumM15,Meeks16}) the term ``monotone'' is used for properties that are closed under the deletion (or addition) of edges only. The latter will be called \emph{edge-monotone} in this work.}~\cite{RothSW20}.

There are two natural generalisations of the class of monotone properties: 
\begin{enumerate}
    \item Properties that are closed under vertex-deletion, called \emph{hereditary} properties (these properties are closed under taking induced subgraphs).
    \item Properties that are closed under edge-deletion, called \emph{edge-monotone} properties.
\end{enumerate}

To make this concrete, let us give some simple examples that are covered by the different notions of closure.
Say $\Phi$ corresponds to the property of being ``planar''. Then $\Phi$ is closed under both vertex- and edge-deletion as well as under edge-contraction. Therefore, $\Phi$ is minor-closed and hence the complexity of $\#\indsubsprob(\Phi)$ is covered by~\cite{JerrumM15density}. The property of being ``bipartite'' is also closed under both vertex- and edge-deletion. However, edge-contractions can lead to non-bipartite graphs. So this is an example of a property that is monotone but not minor-closed, and the corresponding hardness result is from~\cite{RothSW20}. 
An example for a hereditary property that is not monotone is the property ``claw-free'', which refers to the absence of an induced claw. This property is closed under vertex-deletion, but not under edge-deletion.
Conversely, the property of being ``disconnected'' is closed under edge-deletion but not under vertex-deletion. Hence, it is an example for an edge-monotone property that is not monotone (which implies that it is also not hereditary).

So far, there are only partial results on resolving Conjecture~\ref{conj:main} for the hereditary and edge-monotone cases; see~\cite[Section 6]{RothSW20} for hereditary properties defined by a single forbidden induced subgraph (which includes the property ``claw-free''), and~\cite{Meeks16,RothS18} for some results on edge-monotone properties (including the property ``disconnected''). 
Many natural properties are covered by these partial results but a full classification has remained elusive.
In this work, we obtain a full classification for the first case, i.e., for all hereditary properties; our results are presented in Section~\ref{sec:ourResults}.
At this point let us give an example of a hereditary property that, as far as we are aware, has not been covered by previous work. Consider $\Phi$ with $\Phi(H)=1$ iff $H$ is ``hole-free'', which means that $H$ does not have an induced cycle of length at least $5$ (also known as a \emph{hole}).
First note, that $\Phi$ is closed under vertex-deletion, but not under edge-deletion. It is therefore hereditary but not monotone (or even minor-closed).
Since $\Phi$ is characterised by multiple forbidden induced subgraphs it is not covered by~\cite[Theorem 4]{RothSW20}. As $\Phi$ does not distinguish bicliques from independent sets it is not subject to~\cite[Theorem 2]{RothSW20}. Finally, $\Phi$ does not have low Hamming-weight $f$-vectors, which is another criterion introduced in~\cite{RothSW20}. (For this fact, it is relevant that triangles are not forbidden by $\Phi$.)
Similar hereditary properties that have not been covered by previous work are ``(odd-hole)-free'', ``(anti-hole)-free'', etc.
In Section~\ref{sec:conclusions}, we give an example of an unresolved edge-monotone property.

\subsection{Our Results}\label{sec:ourResults}
In addition to confirming Conjecture~\ref{conj:main} for hereditary properties, we also establish a tight conditional lower bound under the Exponential Time Hypothesis; it turns out that a hereditary property is \meagre if and only if either $\Phi$ is true for all graphs, or it is true only for finitely many graphs (for more details, see Observation~\ref{obs:hereditarynontriv}).

\begin{restatable}{theorem}{hereditary}\label{thm:hereditary}
Let $\Phi$ be a computable hereditary graph property. 
If $\Phi$ is \meagre, then $\#\indsubsprob(\Phi)$ is solvable in polynomial time. 
Otherwise $\#\indsubsprob(\Phi)$ is $\#\W{1}$-complete and, assuming the Exponential Time Hypothesis (ETH), cannot be solved in time $f(k)\cdot |G|^{o(k)}$ for any function $f$.
\end{restatable}
Observe that our conditional lower bound under ETH rules out any significant improvement over the brute-force algorithm for $\#\indsubsprob(\Phi)$, which iterates over every $k$-vertex subset of $V(G)$ and counts those that induce a subgraph satisfying $\Phi$. The running time of this algorithm is clearly bounded by $f(k)\cdot |V(G)|^{k+O(1)} \leq f(k) \cdot |G|^{O(k)}$ for some computable function $f$. Note further, that a stronger lower bound ruling out algorithms running in time $f(k)\cdot |V(G)|^{k-\varepsilon}$ for any $\varepsilon>0$ is not possible: For the (hereditary) property $\Phi$ of being a complete graph, the problem $\#\indsubsprob(\Phi)$ is the problem of counting $k$-cliques, which can be solved in time $|V(G)|^{\frac{\omega k}{3}+O(1)}$, where $\omega<3$ is the matrix multiplication exponent~\cite{NesetrilP85}.

To compare our results on exact counting with the complexity of decision and approximate counting we partition the class of all hereditary properties as follows; we write $I_\ell$ and $K_\ell$ for the independent set and the complete graph of size $\ell$, respectively.
\begin{enumerate}
    \item Suppose there are positive integers $s$ and $t$ such that $\Phi$ is false on $K_s$ and $I_t$. By Ramsey's Theorem and the fact that $\Phi$ is closed under taking induced subgraphs, $\Phi$ must then be false on all but finitely many graphs. The problem $\#\indsubsprob(\Phi)$ is thus solvable in polynomial time, and so are its decision and approximate counting versions.
    \item If, for all positive integers $\ell$, the property $\Phi$ is true on $K_\ell$ and $I_\ell$ then $\Phi$ is not \meagre, unless it is constant true. Khot and Raman~\cite{KhotR02} proved that deciding the existence of a $k$-vertex induced subgraph that satisfies $\Phi$ is fixed-parameter tractable. Furthermore, Meeks~\cite{Meeks16} established the existence of a ``fixed-parameter tractable approximation scheme'' (FPTRAS) for the counting problem, which can be considered the parameterised notion of an FPRAS (see~\cite{ArvindR02} for a discussion). 
    
    In sharp contrast, from Theorem~\ref{thm:hereditary} it follows that exact counting is intractable (unless $\Phi$ is trivially true, in which case $\#\indsubsprob(\Phi)$ is trivial).
    \item Otherwise, the decision version was shown to be $\W{1}$-hard by Khot and Raman~\cite{KhotR02}. However, their reduction only yields an implicit ETH-based conditional lower bound of the form $f(k)\cdot |G|^{o(k^{c})}$, where $0<c<1$ is a constant depending on the set of forbidden induced subgraphs of $\Phi$.\footnote{That is, a conditional lower bound that applies to all $\Phi$ is of the form $f(k)\cdot|G|^{o(g(k))}$ where $g$ is asymptotically smaller than every proper rational function, e.g., $g(k)=\log(k)$.} While it is unsurprising that Theorem~\ref{thm:hereditary} yields $\#\W{1}$-hardness of exact counting (since the decision version is $\W{1}$-hard), it is worth to point out that our conditional lower bound significantly improves upon the hardness of decision. 
\end{enumerate}

In summary, together with Khot and Raman~\cite{KhotR02}, and Meeks~\cite{Meeks16}, we fully complete the complexity landscape for detection, approximate counting and exact counting induced subgraphs with hereditary properties. In particular, we identify a significant variety of properties for which decision and approximate counting is easy, but exact counting is hard. A more concise overview is given in Table~\ref{tab:results}. 

  \begin{table}[t]\small
    \begin{tabularx}{\textwidth}{lccc}
        \toprule
        \textbf{Condition on} $\Phi\neq \mathsf{True}^\dagger$
          &\begin{tabular}{c}
             \textbf{Decision}  \\
             Khot \& Raman~\cite{KhotR02}
        \end{tabular}& \begin{tabular}{c}
             \textbf{Approx.\ Counting}   \\
             Meeks~\cite{Meeks16}
        \end{tabular}& \begin{tabular}{c}
             \textbf{Exact Counting}   \\
             This work (Theorem~\ref{thm:hereditary})
        \end{tabular}\\
         \midrule
        $\exists\, s, t: \Phi(K_s) = \Phi(I_t)= 0^\ddagger$
        & 
        P
        &
        P
        &
        P
        \\
        \midrule
             $\forall\, \ell: \Phi(K_\ell) = \Phi(I_\ell)=1$
        
         & \begin{tabular}{c}
         FPT
         \end{tabular} &
         \begin{tabular}{c}
        FPTRAS
        \end{tabular}&
         \begin{tabular}{c}
        $\#\W{1}$-hard,\\
        not in $f(k)\cdot |G|^{o(k)}$
        \end{tabular}\\
        \midrule
        Otherwise
         & \begin{tabular}{c}
         $\W{1}$-hard
         \end{tabular} &
         \begin{tabular}{c}
        no FPTRAS
        \end{tabular}&
         \begin{tabular}{c}
        $\#\W{1}$-hard,\\
        not in $f(k)\cdot |G|^{o(k)}$
        \end{tabular}\\
       
        \bottomrule
     \end{tabularx}
     \small \caption[caption]{Finding and counting $k$-vertex induced subgraphs that satisfy a hereditary property $\Phi$. The conditional lower bounds for exact counting assume the Exponential Time Hypothesis, and the absence of an FPTRAS for approximate counting is conditioned on the assumption that $\W{1}$ does not coincide with FPT under randomised parameterised reductions.\\
     
     \footnotesize{$^\dagger$ For the property $\mathsf{True}$ of being constant true, all versions of the problem become trivial.}
     \\
     \footnotesize{$^\ddagger$ By Ramsey's Theorem, and since $\Phi$ is hereditary, the condition $\exists\, s,t: \Phi(K_s) = \Phi(I_t)= 0$ implies that $\Phi$ is false for all graphs with at least $R(s,t)$ vertices. As $s$ and $t$ are constants, all versions of the problem can be trivially solved in time $n^{O(R(s,t))}=n^{O(1)}$.}
     }
     \label{tab:results}
\end{table}

We note that our classification for hereditary properties subsumes and strengthens the classification of monotone properties due to Roth, Schmitt and Wellnitz~\cite{RothSW20}\footnote{However, the classification of properties depending only on the number of edges in~\cite{RothSW20} is not subsumed.}, see Section~\ref{sec:furtherWork}.

In the course of establishing Theorem~\ref{thm:hereditary}, we prove a much stronger technical intractability theorem which we state in full generality in Section~\ref{sec:technical}. We have not yet explored the full extent of its applicability and we believe that it will be useful in future work (see Section~\ref{sec:conclusions}). For now, let us present one particular additional consequence: We say that a graph property $\Phi$ is \emph{invariant under homomorphic equivalence} if $\Phi(H_1)=\Phi(H_2)$ whenever $H_1$ and $H_2$ are homomorphically equivalent, i.e., there are homomorphisms from $H_1$ to $H_2$ and from $H_2$ to $H_1$. Examples of properties invariant under homomorphic equivalence include
\begin{itemize}
    \item $\Phi(H)=1$ if and only if $H$ has odd girth $d$.
    \item $\Phi(H)=1$ if and only if $H$ has clique number $d$.
    \item $\Phi(H)=1$ if and only if $H$ has chromatic number $d$.
\end{itemize}
Here, $d$ can be any fixed positive integer. We note that none of the previous works on $\#\indsubsprob(\Phi)$ reveals its complexity for the previous three properties. We change that in the current work:

\begin{restatable}{theorem}{homclosed}\label{thm:homequivalence}
Let $\Phi$ be a computable graph property that is invariant under homomorphic equivalence. 
If $\Phi$ is \meagre, then $\#\indsubsprob(\Phi)$ is solvable in polynomial time. 
Otherwise, $\#\indsubsprob(\Phi)$ is $\#\W{1}$-complete and, assuming ETH, cannot be solved in time $f(k)\cdot |G|^{o(k)}$ for any function $f$.
\end{restatable}
As a consequence, for each $d\geq 1$ (and $d$ odd in the case of odd girth), all three of the previous examples yield intractability of $\#\indsubsprob(\Phi)$.

\subsection{Technical Overview}\label{sec:techniques}
Similarly as in previous work~\cite{CurticapeanDM17,RothS18,DorflerRSW22,RothSW20} we rely on the framework of graph motif parameters and the homomorphism basis as introduced by Curticapean, Dell and Marx~\cite{CurticapeanDM17}: Writing $\#\indsubs{\Phi,k}{G}$ for the number of $k$-vertex induced subgraphs of $G$ that satisfy $\Phi$, it is known that there is a unique function $a_{\Phi,k}$ with finite support and independent from $G$, such that
\begin{equation}\label{eq:hom_basis}
    \#\indsubs{\Phi,k}{G} = \sum_H a_{\Phi,k}(H)\cdot \#\homs{H}{G}\,,
\end{equation}
where the sum is over all (isomorphism types of) graphs, and $\#\homs{H}{G}$ denotes the number of graph homomorphisms from $H$ to $G$. Let us emphasise that the sum is finite, since $a_{\Phi,k}$ has finite support, that is, $a_{\Phi,k}(H)\neq 0$ only for finitely many $H$. The \emph{complexity monotonicity principle}, which was independently discovered by Curticapean, Dell and Marx~\cite{CurticapeanDM17} and by Chen and Mengel~\cite{ChenM16}, states that computing a finite linear combination of homomorphism counts as in \eqref{eq:hom_basis} is precisely as hard as computing its hardest term $\#\homs{H}{G}$ with a non-zero coefficient $a_{\Phi,k}(H)\neq 0$. Since the complexity of counting homomorphisms from~$H$ to~$G$ is well-understood --- the problem is feasible if and only if the treewidth\footnote{Intuitively, treewidth is a parameter that measures how tree-like a graph is. In this work, we will rely on treewidth purely in a black-box manner, and thus we refer the reader to Chapter~7 in~\cite{CyganFKLMPPS15} for a comprehensive treatment of treewidth.} of $H$ is small~\cite{DalmauJ04} --- the complexity monotonicity principle shifted the study of the complexity of $\#\indsubsprob(\Phi)$ and related subgraph counting problems to the purely combinatorial problem of determining the treewidth of the graphs~$H$ with a non-zero coefficient $a_{\Phi,k}(H)\neq 0$. More formally, we can define a function $t_\Phi$ which maps a positive integer~$k$ to the maximum treewidth of a graph $H$ with $a_{\Phi,k}(H)\neq 0$. We then obtain the following \emph{implicit} classification:
\begin{theorem}[Corollary 1.11 in~\cite{CurticapeanDM17}]\label{thm:intro_monotonicity}
  The problem $\#\indsubsprob(\Phi)$ is fixed-parameter tractable if $t_\Phi$ is bounded by a constant, and $\#\W{1}$-complete otherwise.
\end{theorem}
\noindent For tight(er) lower bounds under ETH, it is additionally necessary that $t_\Phi(k)\in \Omega(k)$.

With Theorem~\ref{thm:intro_monotonicity} as a powerful tool at hand, recent work focused on establishing an \emph{explicit} criterion for tractability of $\#\indsubsprob(\Phi)$. More concretely, we note that Conjecture~\ref{conj:main} can be resolved if it is proved that $t_\Phi$ is bounded if and only if $\Phi$ is \meagre. Unfortunately, it turned out that the analysis of $t_\Phi$ and thus the analysis of the coefficients $a_{\Phi,k}(H)$ in~\eqref{eq:hom_basis} is a very challenging task in its own right. The reason for the latter is that the coefficients $a_{\Phi,k}(H)$ often encode algebraic and even topological invariants.\footnote{As a concrete example, it was shown in~\cite{RothS18} that for edge-monotone $\Phi$, the coefficient $a_{\Phi,k}(K_k)$ of the complete graph is equal to the so-called reduced Euler characteristic of the simplicial graph complex associated with $\Phi$. As a consequence, it was established that $a_{\Phi,k}(K_k)\neq 0$ is a sufficient criterion for the property $\Phi$ to be \emph{evasive} on $k$-vertex graphs. Therefore, a proof that $a_{\Phi,k}(K_k)$ does not vanish whenever $\Phi$ is non-trivial on $k$-vertex graphs would resolve Karp's famous Evasiveness-Conjecture. We refer the reader to~\cite{RothS18} for a detailed treatment of the connection between the coefficients $a_{\Phi,k}(K_k)$ and the evasiveness of $\Phi$.} Despite the latter difficulty, Theorem~\ref{thm:intro_monotonicity} was successfully used in previous works to resolve Conjecture~\ref{conj:main} for some restricted classes of properties, which we will present in more detail in Section~\ref{sec:conclusions}.

In the current work, we side-step the problem of analysing the coefficients in~\eqref{eq:hom_basis} for hereditary properties by considering a bipartite version of $\#\indsubsprob(\Phi)$ as an intermediate step. For the definition of the intermediate problem, we need to consider bipartite graphs $G$ with \emph{fixed} bipartitions $V(G)=U\dot\cup V$ (note that a bipartite graph might have multiple bipartitions). Formally, we will write $\bG=(U,V,E)$ to emphasise fixing the left- and right-hand side vertices $U$ and $V$, respectively. Furthermore, two bipartite graphs $\bG_1=(U_1,V_1,E_2)$ and $\bG_2=(U_2,V_2,E_2)$ are said to be \emph{consistently isomorphic} if there exists an isomorphism that maps $U_1$ to $U_2$ and $V_1$ to $V_2$, respectively. A \emph{bipartite property} $\Psi$ is then defined to be a function from bipartite graphs to $\{0,1\}$ such that $\Psi(\bG_1)=\Psi(\bG_2)$ whenever $\bG_1$ and $\bG_2$ are consistently isomorphic. 

Given a bipartite property $\Psi$, the problem $\#\bipindsubsprob(\Psi)$ asks, on input a bipartite graph~$\bG$ (with fixed bipartition!) and a positive integer $k$, to compute the number of $k$-vertex induced subgraphs of $\bG$ that satisfy $\Psi$; here, the bipartition of an induced subgraph of $\bG$ is induced by the bipartition of $\bG$. We stress that $\#\bipindsubsprob(\Psi)$ is \emph{not} the same as the restriction of $\#\indsubsprob(\Phi)$ to bipartite input graphs (without fixed bipartition). For example, $\#\bipindsubsprob(\Psi)$ allows us to express counting of $2k$-vertex induced subgraphs of $\bG$ that have $k$ vertices on the left-hand side and $k$ vertices on the right-hand side, or, more interestingly, counting $k$-vertex induced subgraphs of $\bG$ such that there is a vertex on the left-hand side that is adjacent to all vertices on the right-hand side. Both of those examples are not expressible by just restricting $\#\indsubsprob(\Phi)$ to bipartite inputs.

Our proof of Theorem~\ref{thm:hereditary} can then be split into two essentially independent parts: First, we establish the following criterion for the intractability of $\#\bipindsubsprob(\Psi)$. To this end, $\bI_{k,k}$ denotes an independent set of size $2k$, with a fixed bipartition that contains $k$ vertices on the left-hand side and $k$ vertices on the right-hand side; and $\bB_{k,k}$ denotes the complete bipartite graph with $k$ vertices on the left-hand side and $k$ vertices on the right-hand side. Furthermore, we call a set of integers $\mathcal{K}$ \emph{dense} if there exists a constant $c$ such that for every positive integer $m$, there is a $k\in  \mathcal{K}$ with $m\leq k\leq c\cdot m$.
\begin{restatable}{theorem}{mainbipartite}\label{thm:main_bipartite}
Let $\Psi$ be a computable bipartite property.
Let $\mathcal{K}$ be the set of primes $k$ for which $\Psi$ distinguishes $\bI_{k,k}$ and $\bB_{k,k}$, i.e., $\Psi(\bI_{k,k})\neq \Psi(\bB_{k,k})$.
If $\mathcal{K}$ is infinite then $\#\bipindsubsprob(\Psi)$ is $\#\W{1}$-hard.
Moreover, if $\mathcal{K}$ is dense then $\#\bipindsubsprob(\Psi)$ cannot be solved in time $f(k)\cdot |G|^{o(k)}$ for any function $f$, assuming the ETH.
\end{restatable}

In the second step, we show that for a wide range of properties $\Phi$, including all hereditary properties, we can associate with $\Phi$ a bipartite property $\Psi_\Phi$ such that $\#\bipindsubsprob(\Psi_\Phi)$ reduces to $\#\indsubsprob(\Phi)$ with respect to parameterised reductions. Additionally, this reduction will be tight in the sense that all conditional lower bounds transfer. Finally, we show that, whenever $\Phi$ is not \meagre, the bipartite property $\Psi_\Phi$ will satisfy the strong hardness condition in Theorem~\ref{thm:main_bipartite}, yielding not only $\#\W{1}$-hardness, but also the conditional lower bound under ETH.
\footnote{For readers familiar with the so-called bipartite double-cover $G\times K_2$, we wish to stress that the latter is \emph{not} used in our reduction, even though this approach may seem tempting at first glance. Unfortunately, due to technical reasons which are out of the scope of this extended abstract, we were not able to obtain an easier proof via the bipartite double-cover.} 

In what follows, we will describe both steps in more detail separately.

\subsubsection{Classification of Bipartite Properties}
The main motivation of our consideration of $\#\bipindsubsprob(\Psi)$ as an intermediate step is the ``algebraic approach to hardness'' as introduced in~\cite{DorflerRSW22}, which we will describe subsequently.

For a properly defined vertex-coloured version of $\#\indsubsprob(\Phi)$, restricted to bipartite input graphs (but without fixed bipartition), a transformation as linear combination of vertex-coloured homomorphism counts similar as in Equation~\eqref{eq:hom_basis} is known to hold. It was furthermore established that 
\begin{enumerate}
    \item[(I)] the complexity monotonicity principle (Theorem~\ref{thm:intro_monotonicity}) remains true in the vertex-coloured setting, and
    \item[(II)] the coefficient $a_{\Phi,k}(H)$ for $H$ being a complete bipartite graph can be analysed much easier in the vertex-coloured case.
\end{enumerate}
The reason for the simplified analysis of the coefficient in (II) ultimately relied on the fact that the complete bipartite graph is edge-transitive and that it can have a prime-power number of edges; we describe this in more detail when we apply the algebraic approach to the setting of fixed bipartitions further below.

In combination, (I) and (II) were shown to yield a classification similar to Theorem~\ref{thm:main_bipartite} but without considering fixed bipartitions.
Unfortunately, our reduction from the bipartite to the non-bipartite case crucially depends on such fixed bipartitions. Therefore, we adapt the algebraic approach to $\#\bipindsubsprob(\Psi)$ as follows. 
Writing $\#\bipindsubs{\Psi,k}{\bG}$ for the number of (bipartite) $k$-vertex induced subgraphs of $\bG$ that satisfy~$\Psi$, we establish a similar transformation as in Equation~\eqref{eq:hom_basis}. We show that there exists a function $a_{\Psi,k}$ of finite support and independent of $\bG$ such that\footnote{In fact, for technical reasons, we establish Equation~\eqref{eq:hom_basis_bipartised} in a vertex-coloured setting, which is, however, shown to be interreducible with the uncoloured setting. The formal statement is given by Lemma~\ref{lem:bipcm}.}
\begin{equation}\label{eq:hom_basis_bipartised}
    \#\bipindsubs{\Phi,k}{\bG} = \sum_H a_{\Psi,k}(H) \cdot \#\homs{H}{G} \,, 
\end{equation}
where the sum is again over all graphs (without fixed bipartitions) and $G$ is the underlying graph of $\bG$. Additionally, if $k=2\ell$, we show that
\begin{equation}\label{eq:intro_coeff}
    a_{\Psi,k}(B_{\ell,\ell})= \sum_{A\subseteq E(B_{\ell,\ell})} \Psi(\bB_{\ell,\ell}[A]) \cdot (-1)^{\ell^2-|A|}\,,
\end{equation}
where $B_{\ell,\ell}$ is the complete bipartite graph, i.e., the $\ell$-by-$\ell$ biclique, $\bB_{\ell,\ell}$ is the $\ell$-by-$\ell$ biclique with fixed bipartition, and $\bB_{\ell,\ell}[A]$ is obtained from $\bB_{\ell,\ell}$ by removing all edges in $E(B_{\ell,\ell})\setminus A$. The goal is to show that $a_{\Psi,k}(B_{\ell,\ell})$ in~\eqref{eq:intro_coeff} is non-zero whenever $\ell$ is a prime and $\Psi(\bI_{\ell,\ell})\neq\Psi(\bB_{\ell,\ell})$. Since the treewidth of $B_{\ell,\ell}$ is linear in $\ell$, we can rely on a vertex-coloured version of complexity monotonicity such as in (I) to prove that $a_{\Psi,k}(B_{\ell,\ell})\neq 0$ is sufficient for the classification of $\#\bipindsubsprob(\Psi)$ (Theorem~\ref{thm:main_bipartite}).

The subtle difference between~\eqref{eq:intro_coeff} and the analysis in~\cite{DorflerRSW22}, which prevents us from using the main result of~\cite{DorflerRSW22} in a black-box manner, is that the edge-subgraphs of $\bB_{\ell,\ell}$ keep their fixed bipartition, and only the subgraphs for which the bipartite property $\Psi$ holds contribute to the sum. More precisely, the main result of~\cite{DorflerRSW22} is achieved by considering the canonical action of the automorphism group $\auts{B_{\ell,\ell}}$ on the set of edge-subsets and observing that the term $\Psi(B_{\ell,\ell}[A]) \cdot (-1)^{\ell^2-|A|}$ is invariant under this action if no bipartition is fixed and if $\Psi$ is a graph property rather than a bipartite property. However, in our case, $\Psi$ is a bipartite property that respects the bipartition. Thus, for an automorphism $\pi$ of $B_{\ell,\ell}$ which maps vertices from the left-hand side to the right-hand side and vice versa, there might be an edge-subset $A$ such that
\[\Psi(\bB_{\ell,\ell}[A]) \neq \Psi(\bB_{\ell,\ell}[\pi(A)]) \,.\]

Fortunately, we can easily solve this problem by restricting to automorphisms that are \emph{consistent}, i.e., which map the left-hand side to the left-hand side and the right-hand side to the right-hand side. Writing $\auts{\bB_{\ell,\ell}}$ for the set of consistent automorphisms, it is easy to see that $\auts{\bB_{\ell,\ell}}$ still acts transitively on the edges of the complete bipartite graph $\bB_{\ell,\ell}$, that is, for each pair of edges $e$ and $f$ of $\bB_{\ell,\ell}$, there exists $\pi\in \auts{\bB_{\ell,\ell}}$ such that $\pi(e)=f$. 
With that observation at hand, we can apply the algebraic approach similarly as in~\cite{DorflerRSW22}; for now we provide a concise outline and refer the reader to Section~\ref{sec:bipartite} for the detailed presentation.

To establish that $a_{\Psi,k}(B_{\ell,\ell})$ does not vanish under the previous constraints, we will first observe that $\#\auts{\bB_{\ell,\ell}}$ is divisible by $\ell$. Thus, we can show that there exists an $\ell$-Sylow subgroup $\Gamma$ of $\auts{\bB_{\ell,\ell}}$ whose action on the edges of $\bB_{\ell,\ell}$ is still transitive. Extending this action to edge-subsets of $\bB_{\ell,\ell}$, we observe that for each pair $A_1$ and $A_2$ in the same orbit, we have that
\[\Psi(\bB_{\ell,\ell}[A_1]) = \Psi(\bB_{\ell,\ell}[A_2]) \,.\]
Since the size of each orbit must divide the order of the group, which is the prime $\ell$, we can take Equation~\eqref{eq:intro_coeff} modulo $\ell$, and observe that only the fixed points survive, that is $A=\emptyset$ and $A=E(B_{\ell,\ell})$. In other words, we obtain
\begin{equation}\label{eq:intro_coeff_modell}
    a_{\Psi,k}(B_{\ell,\ell})= \Psi(\bB_{\ell,\ell}) + \Psi(\bI_{\ell,\ell}) \cdot (-1)^{\ell^2} \mod \ell\,,
\end{equation}
which is non-zero whenever $\ell$ is a prime and $\Psi(\bB_{\ell,\ell}) \neq \Psi(\bI_{\ell,\ell})$. As outlined previously, this will suffice for establishing the classification of $\#\bipindsubsprob(\Psi)$ (Theorem~\ref{thm:main_bipartite}).

\subsubsection{Reducing from Bipartite Properties to Graph Properties using False Twins}

Let $\Phi$ be a hereditary graph property that is not \meagre. In order to confirm Conjecture~\ref{conj:main} for such $\Phi$ we will relate $\Phi$ to a bipartite property $\Psi_{\Phi}$, which satisfies the requirements of the hardness result from Theorem~\ref{thm:main_bipartite}. Intuitively, this means that $\Psi_\Phi$ should distinguish, for certain $k$, the independent set $\bI_{k,k}$ from the biclique~$\bB_{k,k}$. Additionally, we have to establish that $\#\bipindsubsprob(\Psi_\Phi)$ reduces to $\#\indsubsprob(\Phi)$.

How could we define such a property $\Psi_\Phi$? Since $\Phi$ is hereditary it can be classified by a (possibly infinite) set of (inclusion-minimal) forbidden induced subgraphs $\obstr(\Phi)$. Since $\Phi$ is not \meagre $\obstr(\Phi)$ contains at least one element, say $H$. Consider the following initial construction; an illustration is provided in Figure~\ref{fig:ghat}. Suppose that $H$ contains an edge $e=\{u_1,u_2\}$. Then replacing this edge with a complete bipartite graph $\bB_{k,k}$ yields a graph $H_{\bB_{k,k}}$ that contains $H$ as induced subgraph, which means that $\Phi(H_{\bB_{k,k}})=0$. Now suppose further that replacing $e$ with an independent set $\bI_{k,k}$ yields a graph $H_{\bI_{k,k}}$ for which $\Phi(H_{\bI_{k,k}})=1$. Then the process of replacing the edge $e$ with some bipartite graph $\bG$ --- note that the choice of the bipartition matters --- and evaluating $\Phi$ for the resulting graph $H_{\bG}$ defines a bipartite property that distinguishes $\bI_{k,k}$ from $\bB_{k,k}$.

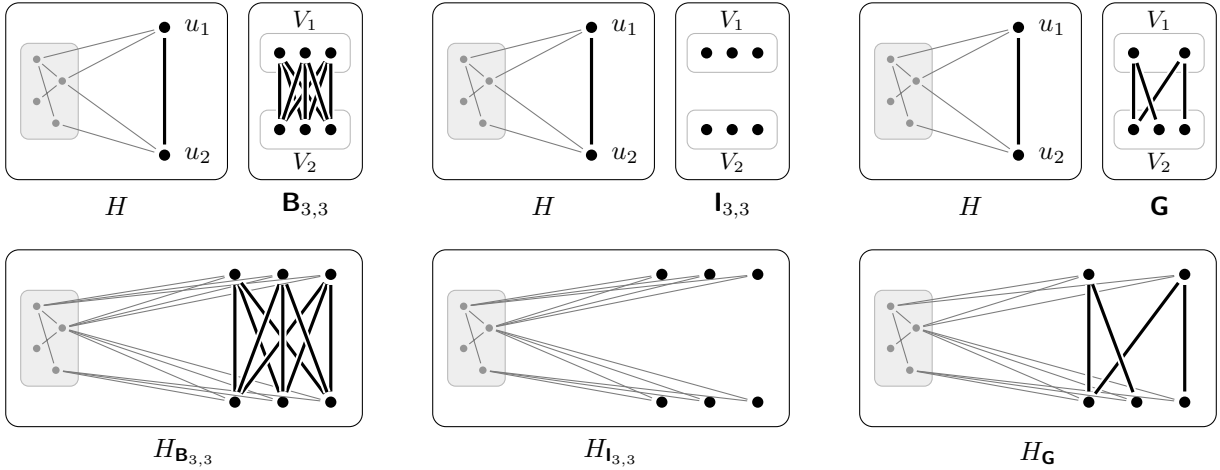
\begin{figure}[t]
        \centering
        \begin{tikzpicture}[scale=0.85]
            \begin{scope}
                \node[vertex,label=right:{$u_1$}] (u) at (0,0) {};
                \node[vertex,label=right:{$u_2$}] (v) at (0,-2) {};

                \draw[edge] (u) -- (v);

                \node (e1) at (-2.2, .1) {};
                \node (e2) at (0.7, -2.1) {};
                \node[cvertex] (x1) at (-1.8, -1.5) {};
                \node[cvertex] (x2) at (-2, -.5) {};
                \node[cvertex] (x3) at (-1.6, -1) {};
                \node[fit=(x1)(x2)(x3), rounded corners=3pt, draw=lipicsGray!40,
                    fill=lipicsGray!10] {};
                \node[fit=(e1)(e2), rounded corners=5pt, draw] (f) {};
                \node at ($(f.south) + (0,-.4)$) {$H$};

                \node[cvertex] (x1) at (-1.7, -1.5) {};
                \node[cvertex] (x2) at (-2, -.5) {};
                \node[cvertex] (x3) at (-1.6, -.84) {};
                \node[cvertex] (x4) at (-2, -1.16) {};
                \draw[cedge] (u) -- (x3) -- (v);
                \draw[cedge] (u) -- (x2);
                \draw[cedge] (v) -- (x1) -- (x2) -- (x3) -- (x4);
                
                \node[vertex] (U2) at (2.2, -0.4) {};
                \node[vertex] (U3) at (2.6, -0.4) {};
                \node[vertex] (U4) at (1.8, -0.4) {};
                \node[vertex] (V2) at (2.6, -1.6) {};
                \node[vertex] (V3) at (2.2, -1.6) {};
                \node[vertex] (V4) at (1.8, -1.6) {};
                \node[fit=(U3)(U4), rounded corners=3pt, draw=lipicsGray!40] (x1) {};
                \node[fit=(V2)(V3)(V4), rounded corners=3pt, draw=lipicsGray!40] (x2) {};
                \node at ($(x2.south) + (0,-.2)$) {\small$V_2$};
                \node at ($(x1.north) + (0,.2)$) {\small$V_1$};
                \node (f1) at (2.81, .1) {};
                \node (f2) at (1.6, -2.1) {};
                \draw[sdedge] (U4) -- (V4);
                \draw[sdedge] (U4) -- (V3);
                \draw[sdedge] (U4) -- (V2);
                \draw[sdedge] (U3) -- (V4);
                \draw[sdedge] (U3) -- (V3);
                \draw[sdedge] (U3) -- (V2);
                \draw[sdedge] (U2) -- (V4);
                \draw[sdedge] (U2) -- (V3);
                \draw[sdedge] (U2) -- (V2);
                \node[fit=(f1)(f2), rounded corners=5pt, draw] (g) {};
                \node at ($(g.south) + (0,-.4)$) {$\bB_{3,3}$};
            \end{scope}

            \begin{scope}[yshift = -11em]
             \node (e1) at (-2.2, .1) {};
                \node (e2) at (.4, -2.1) {};
                \node (f1) at (2.81, .1) {};
                \node (f2) at (1.6, -2.1) {};
                \node[cvertex] (x1) at (-1.8, -1.5) {};
                \node[cvertex] (x2) at (-2, -.5) {};
                \node[cvertex] (x3) at (-1.6, -1) {};
                \node[fit=(x1)(x2)(x3), rounded corners=3pt, draw=lipicsGray!40,
                    fill=lipicsGray!10] {};
                \node[fit=(e1)(e2)(f1)(f2), rounded corners=5pt, draw] (f) {};
                \node at ($(f.south) + (0,-.4)$) {$H_{\bB_{3,3}}$};

                \node[cvertex] (x1) at (-1.7, -1.5) {};
                \node[cvertex] (x2) at (-2, -.5) {};
                \node[cvertex] (x3) at (-1.6, -.84) {};
                \node[cvertex] (x4) at (-2, -1.16) {};

                \node[vertex] (U3) at (2.6, -0) {};
                \node[vertex] (U4) at (1.1, -0) {};
                \node[vertex] (U2) at (1.85, -0) {};
                \node[vertex] (V2) at (2.6, -2) {};
                \node[vertex] (V3) at (1.85, -2) {};
                \node[vertex] (V4) at (1.1, -2) {};
                \draw[cedge] (U4) -- (x3) -- (V4);
                \draw[cedge] (U3) -- (x3) -- (V2);
                \draw[cedge] (x3) -- (V3);
                \draw[cedge] (U3) -- (x2);
                \draw[cedge] (U4) -- (x2);
                \draw[cedge] (V4) -- (x1) -- (x2) -- (x3) -- (x4);
                \draw[cedge] (V2) -- (x1);
                \draw[cedge] (V3) -- (x1);
                \draw[cedge] (U2) -- (x3);
                \draw[cedge] (U2) -- (x2);
                
                \draw[sdedge] (U4) -- (V4);
                \draw[sdedge] (U4) -- (V3);
                \draw[sdedge] (U4) -- (V2);
                \draw[sdedge] (U3) -- (V4);
                \draw[sdedge] (U3) -- (V3);
                \draw[sdedge] (U3) -- (V2);
                \draw[sdedge] (U2) -- (V4);
                \draw[sdedge] (U2) -- (V3);
                \draw[sdedge] (U2) -- (V2);

            \end{scope}
            
             \begin{scope}[xshift=19em]
                 \node[vertex,label=right:{$u_1$}] (u) at (0,0) {};
                \node[vertex,label=right:{$u_2$}] (v) at (0,-2) {};

                \draw[edge] (u) -- (v);

                \node (e1) at (-2.2, .1) {};
                \node (e2) at (.7, -2.1) {};
                \node[cvertex] (x1) at (-1.8, -1.5) {};
                \node[cvertex] (x2) at (-2, -.5) {};
                \node[cvertex] (x3) at (-1.6, -1) {};
                \node[fit=(x1)(x2)(x3), rounded corners=3pt, draw=lipicsGray!40,
                    fill=lipicsGray!10] {};
                \node[fit=(e1)(e2), rounded corners=5pt, draw] (f) {};
                \node at ($(f.south) + (0,-.4)$) {$H$};

                \node[cvertex] (x1) at (-1.7, -1.5) {};
                \node[cvertex] (x2) at (-2, -.5) {};
                \node[cvertex] (x3) at (-1.6, -.84) {};
                \node[cvertex] (x4) at (-2, -1.16) {};
                \draw[cedge] (u) -- (x3) -- (v);
                \draw[cedge] (u) -- (x2);
                \draw[cedge] (v) -- (x1) -- (x2) -- (x3) -- (x4);

                \node[vertex] (U3) at (2.6, -0.4) {};
                \node[vertex] (U4) at (1.8, -0.4) {};
                \node[vertex] (V2) at (2.6, -1.6) {};
                \node[vertex] (U2) at (2.2, -0.4) {};
                \node[vertex] (V3) at (2.2, -1.6) {};
                \node[vertex] (V4) at (1.8, -1.6) {};
                \node[fit=(U3)(U4), rounded corners=3pt, draw=lipicsGray!40] (x1) {};
                \node[fit=(V2)(V3)(V4), rounded corners=3pt, draw=lipicsGray!40] (x2) {};
                \node at ($(x2.south) + (0,-.2)$) {\small$V_2$};
                \node at ($(x1.north) + (0,.2)$) {\small$V_1$};
                \node (f1) at (2.81, .1) {};
                \node (f2) at (1.6, -2.1) {};
                \node[fit=(f1)(f2), rounded corners=5pt, draw] (g) {};
                \node at ($(g.south) + (0,-.4)$) {$\bI_{3,3}$};

            \begin{scope}[yshift = -11em]
             \node (e1) at (-2.2, .1) {};
                \node (e2) at (.4, -2.1) {};
                \node (f1) at (2.81, .1) {};
                \node (f2) at (1.6, -2.1) {};
                \node[cvertex] (x1) at (-1.8, -1.5) {};
                \node[cvertex] (x2) at (-2, -.5) {};
                \node[cvertex] (x3) at (-1.6, -1) {};
                \node[fit=(x1)(x2)(x3), rounded corners=3pt, draw=lipicsGray!40,
                    fill=lipicsGray!10] {};
                \node[fit=(e1)(e2)(f1)(f2), rounded corners=5pt, draw] (f) {};
                \node at ($(f.south) + (0,-.4)$) {$H_{\bI_{3,3}}$};

                \node[cvertex] (x1) at (-1.7, -1.5) {};
                \node[cvertex] (x2) at (-2, -.5) {};
                \node[cvertex] (x3) at (-1.6, -.84) {};
                \node[cvertex] (x4) at (-2, -1.16) {};

                \node[vertex] (U3) at (2.6, -0) {};
                \node[vertex] (U2) at (1.85, -0) {};
                \node[vertex] (U4) at (1.1, -0) {};
                \node[vertex] (V2) at (2.6, -2) {};
                \node[vertex] (V3) at (1.85, -2) {};
                \node[vertex] (V4) at (1.1, -2) {};
                \draw[cedge] (U4) -- (x3) -- (V4);
                \draw[cedge] (U3) -- (x3) -- (V2);
                \draw[cedge] (x3) -- (V3);
                \draw[cedge] (U3) -- (x2);
                \draw[cedge] (U4) -- (x2);
                \draw[cedge] (V4) -- (x1) -- (x2) -- (x3) -- (x4);
                \draw[cedge] (V2) -- (x1);
                \draw[cedge] (V3) -- (x1);
                \draw[cedge] (U2) -- (x3);
                \draw[cedge] (U2) -- (x2);
            
         \end{scope}
            \end{scope}
             \begin{scope}[xshift=38em]
                 \node[vertex,label=right:{$u_1$}] (u) at (0,0) {};
                \node[vertex,label=right:{$u_2$}] (v) at (0,-2) {};

                \draw[edge] (u) -- (v);

                \node (e1) at (-2.2, .1) {};
                \node (e2) at (.7, -2.1) {};
                \node[cvertex] (x1) at (-1.8, -1.5) {};
                \node[cvertex] (x2) at (-2, -.5) {};
                \node[cvertex] (x3) at (-1.6, -1) {};
                \node[fit=(x1)(x2)(x3), rounded corners=3pt, draw=lipicsGray!40,
                    fill=lipicsGray!10] {};
                \node[fit=(e1)(e2), rounded corners=5pt, draw] (f) {};
                \node at ($(f.south) + (0,-.4)$) {$H$};

                \node[cvertex] (x1) at (-1.7, -1.5) {};
                \node[cvertex] (x2) at (-2, -.5) {};
                \node[cvertex] (x3) at (-1.6, -.84) {};
                \node[cvertex] (x4) at (-2, -1.16) {};
                \draw[cedge] (u) -- (x3) -- (v);
                \draw[cedge] (u) -- (x2);
                \draw[cedge] (v) -- (x1) -- (x2) -- (x3) -- (x4);

                \node[vertex] (U3) at (2.6, -0.4) {};
                \node[vertex] (U4) at (1.8, -0.4) {};
                \node[vertex] (V2) at (2.6, -1.6) {};
                \node[vertex] (V3) at (2.2, -1.6) {};
                \node[vertex] (V4) at (1.8, -1.6) {};
                \node[fit=(U3)(U4), rounded corners=3pt, draw=lipicsGray!40] (x1) {};
                \node[fit=(V2)(V3)(V4), rounded corners=3pt, draw=lipicsGray!40] (x2) {};
                \node at ($(x2.south) + (0,-.2)$) {\small$V_2$};
                \node at ($(x1.north) + (0,.2)$) {\small$V_1$};
                \node (f1) at (2.81, .1) {};
                \node (f2) at (1.6, -2.1) {};
                \draw[sdedge] (U3) -- (V4);
                \draw[sdedge] (U3) -- (V2);
                \draw[sdedge] (V3) -- (U4) -- (V4);
                \node[fit=(f1)(f2), rounded corners=5pt, draw] (g) {};
                \node at ($(g.south) + (0,-.4)$) {$\bG$};

            \begin{scope}[yshift = -11em]
             \node (e1) at (-2.2, .1) {};
                \node (e2) at (.4, -2.1) {};
                \node (f1) at (2.81, .1) {};
                \node (f2) at (1.6, -2.1) {};
                \node[cvertex] (x1) at (-1.8, -1.5) {};
                \node[cvertex] (x2) at (-2, -.5) {};
                \node[cvertex] (x3) at (-1.6, -1) {};
                \node[fit=(x1)(x2)(x3), rounded corners=3pt, draw=lipicsGray!40,
                    fill=lipicsGray!10] {};
                \node[fit=(e1)(e2)(f1)(f2), rounded corners=5pt, draw] (f) {};
                \node at ($(f.south) + (0,-.4)$) {$H_{\bG}$};

                \node[cvertex] (x1) at (-1.7, -1.5) {};
                \node[cvertex] (x2) at (-2, -.5) {};
                \node[cvertex] (x3) at (-1.6, -.84) {};
                \node[cvertex] (x4) at (-2, -1.16) {};

                \node[vertex] (U3) at (2.6, -0) {};
                \node[vertex] (U4) at (1.1, -0) {};
                \node[vertex] (V2) at (2.6, -2) {};
                \node[vertex] (V3) at (1.85, -2) {};
                \node[vertex] (V4) at (1.1, -2) {};
                \draw[cedge] (U4) -- (x3) -- (V4);
                \draw[cedge] (U3) -- (x3) -- (V2);
                \draw[cedge] (x3) -- (V3);
                \draw[cedge] (U3) -- (x2);
                \draw[cedge] (U4) -- (x2);
                \draw[cedge] (V4) -- (x1) -- (x2) -- (x3) -- (x4);
                \draw[cedge] (V2) -- (x1);
                \draw[cedge] (V3) -- (x1);
                \draw[sdedge] (U3) -- (V4);
                \draw[sdedge] (U3) -- (V2);
                \draw[sdedge] (V3) -- (U4) -- (V4);
            
         \end{scope}
            \end{scope}
        \end{tikzpicture}
        \caption{Replacing an edge $e$ of $H$ by the $3$-by-$3$ biclique, the $3$-by-$3$ independent set, and the bipartite graph $\bG$.}\label{fig:ghat}
    \end{figure}

However, for this approach to work in general, it is essential that $\Phi(H_{\bI_{k,k}})=1$, i.e., that $H_{\bI_{k,k}}$ does not contain any forbidden induced subgraph. This suggests some kind of ``minimal'' choice of the graph $H\in \obstr(\Phi)$ in the general case. Consider the graphs in Figure~\ref{fig:K4} as forbidden induced subgraphs. The graph on the left has fewer vertices and edges. However, when replacing any of its edges by $\bI_{2,2}$, we obtain the graph on the right. So it does not suffice to look at the number of vertices or edges alone.

\begin{figure}[tbh]
        \centering
        \begin{tikzpicture}[scale=0.85]
            \centering
            \begin{scope}
            
                \node[vertex] (U1) at (0, -0) {};
                \node[vertex] (U2) at (0, -2) {};
                \node[vertex] (V1) at (-2, 0) {};
                \node[vertex] (V2) at (-2, -2) {};
               
                \draw[sdedge] (U1) -- (V1);
                \draw[sdedge] (U1) -- (V2);
                \draw[sdedge] (U2) -- (V1);
                \draw[sdedge] (U2) -- (V2);
                \draw[sdedge] (U1) -- (U2);
                \draw[sdedge] (V1) -- (V2);

                \node (e1) at (-2.7, .1) {};
                \node (e2) at (0.7, -2.1) {};
               
                \node[fit=(e1)(e2), rounded corners=5pt, draw] (f) {};
                \node at ($(f.south) + (0,-.4)$) {$H_1$};

                \node[vertex] (U1) at (6, 0) {};
                \node[vertex] (U3) at (5.5, -0.25) {};
                \node[vertex] (U2) at (6, -2) {};
                \node[vertex] (U4) at (5.5, -1.75) {};
                \node[vertex] (V1) at (4, 0) {};
                \node[vertex] (V2) at (4, -2) {};
               
                \draw[sdedge] (U1) -- (V1);
                \draw[sdedge] (U2) -- (V1);
                \draw[sdedge] (U3) -- (V1);
                \draw[sdedge] (U4) -- (V1);
                \draw[sdedge] (U1) -- (V2);
                \draw[sdedge] (U2) -- (V2);
                \draw[sdedge] (U3) -- (V2);
                \draw[sdedge] (U4) -- (V2);
                \draw[sdedge] (V1) -- (V2);

                \node (e1) at (3.3, .1) {};
                \node (e2) at (6.7, -2.1) {};
               
                \node[fit=(e1)(e2), rounded corners=5pt, draw] (g) {};
                \node at ($(g.south) + (0,-.4)$) {$H_2$};
            \end{scope}
         \end{tikzpicture}
        \caption{$H_1=K_4$ and $H_2$ which can be obtained from $K_4$ by replacing an edge by $\bI_{2,2}$.}\label{fig:K4}
    \end{figure}
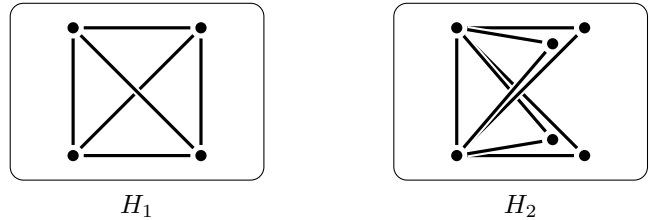

For the choice of $H$ it turns out to be helpful to consider sets of vertices in $H$ that have identical neighbourhood (so-called \emph{false twins}). The false twin relation partitions the vertices of $H$ into \emph{blocks}. By $H\down$ we denote the corresponding quotient graph, in which each block is replaced by a single vertex. We refer to it as the \emph{twin-free quotient}.\footnote{The twin-free quotient was implicitly used in~\cite{RothSW20} as well, although in a much less general reduction.}
Note that the twin-free quotient of the left-hand graph in Figure~\ref{fig:K4} is $K_4$ itself, whereas the twin-free quotient of the right-hand graph is $K_4$ minus an edge.

For our refined construction that will ultimately lead to the definition of the bipartite property $\Psi_\Phi$,
we will choose a graph $H$ from $\obstr(\Phi)$ for which $H\down$ has minimal number of edges. For an edge $\{u_1,u_2\}$ in $H$ let $B_1$ and $B_2$ be the blocks containing $u_1$ and $u_2$, respectively. Given a bipartite graph $\bG=(V_1, V_2, E)$ we define a graph $F_{\bG}$, where we now replace not only the edge $\{u_1,u_2\}$ but the complete bipartite graph induced by $B_1$ and $B_2$, and insert in its stead the graph $\bG$, where $V_1$ replaces $B_1$, and $V_2$ replaces $B_2$, see Figure~\ref{fig:implant} for an example. As in the initial construction, for sufficiently large $k$, $H$ is an induced subgraph of $F_{\bB_{k,k}}$, which implies $\Phi(F_{\bB_{k,k}})=0$. However, one can also show that for every induced subgraph $F'$ of $F_{\bI_{k,k}}$, it holds that $|E(F'\down)|\le |E(F_{\bI_{k,k}}\down)| < |E(H\down)|$, and so, by our choice of $H$, $F'$ is not in $\obstr(\Phi)$, i.e., it is not a forbidden induced subgraph of $\Phi$. Consequently, $\Phi(F_{\bI_{k,k}})=1$ as intended.
This way we establish that the bipartite property $\Psi_\Phi$ with $\Psi_\Phi(\bG)\coloneqq\Phi(F_{\bG})$ distinguishes, for sufficiently large $k$, the independent set $\bI_{k,k}$ from the biclique $\bB_{k,k}$ and thereby satisfies the requirements of Theorem~\ref{thm:main_bipartite}. This shows that $\#\bipindsubsprob(\Psi_\Phi)$ is $\#\W{1}$-hard with the corresponding conditional lower bound.

 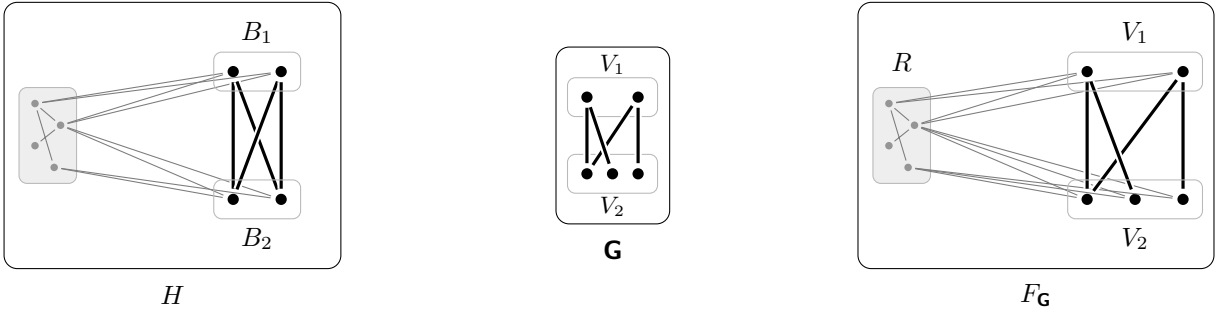
\begin{figure}[t]
        \centering
        \begin{tikzpicture}[scale=0.85]

            \begin{scope}[]
                \node (e1) at (-2.2, .8) {};
                \node (e2) at (.4, -2.8) {};
                \node (f1) at (2.5, .8) {};
                \node (f2) at (1.6, -2.8) {};
                \node[cvertex] (x1) at (-1.8, -1.5) {};
                \node[cvertex] (x2) at (-2, -.5) {};
                \node[cvertex] (x3) at (-1.6, -1) {};
                \node[fit=(x1)(x2)(x3), rounded corners=3pt, draw=lipicsGray!40,
                    fill=lipicsGray!10] {};
                \node[fit=(e1)(e2)(f1)(f2), rounded corners=5pt, draw] (f) {};
                \node at ($(f.south) + (0,-.4)$) {$H$};

                \node[cvertex] (x1) at (-1.7, -1.5) {};
                \node[cvertex] (x2) at (-2, -.5) {};
                \node[cvertex] (x3) at (-1.6, -.84) {};
                \node[cvertex] (x4) at (-2, -1.16) {};

                \node[vertex] (U4) at (1.1, -0) {};
                \node[vertex] (U2) at (1.85, -0) {};
                \node[vertex] (V3) at (1.85, -2) {};
                \node[vertex] (V4) at (1.1, -2) {};
                \draw[cedge] (U4) -- (x3) -- (V4);
                \draw[cedge] (x3) -- (V3);
                \draw[cedge] (U4) -- (x2);
                \draw[cedge] (V4) -- (x1) -- (x2) -- (x3) -- (x4);
                \draw[cedge] (V3) -- (x1);
                \draw[cedge] (U2) -- (x3);
                \draw[cedge] (U2) -- (x2);
                
                \draw[sdedge] (U4) -- (V4);
                \draw[sdedge] (U4) -- (V3);
                \draw[sdedge] (U2) -- (V4);
                \draw[sdedge] (U2) -- (V3);
                
                \node[fit=(U2)(U4), rounded corners=3pt, draw=lipicsGray!40] (x1) {};
                \node[fit=(V3)(V4), rounded corners=3pt, draw=lipicsGray!40] (x2) {};
                \node at ($(x1.north) + (0,.3)$) {$B_1$};
                \node at ($(x2.south) + (0,-.3)$) {$B_2$};
            \end{scope}

             \begin{scope}[xshift=13.75em]
                \node[vertex] (U3) at (2.6, -0.4) {};
                \node[vertex] (U4) at (1.8, -0.4) {};
                \node[vertex] (V2) at (2.6, -1.6) {};
                \node[vertex] (V3) at (2.2, -1.6) {};
                \node[vertex] (V4) at (1.8, -1.6) {};
                \node[fit=(U3)(U4), rounded corners=3pt, draw=lipicsGray!40] (x1) {};
                \node[fit=(V2)(V3)(V4), rounded corners=3pt, draw=lipicsGray!40] (x2) {};
                \node at ($(x2.south) + (0,-.2)$) {\small$V_2$};
                \node at ($(x1.north) + (0,.2)$) {\small$V_1$};
                \node (f1) at (2.81, .1) {};
                \node (f2) at (1.6, -2.1) {};
                \draw[sdedge] (U3) -- (V4);
                \draw[sdedge] (U3) -- (V2);
                \draw[sdedge] (V3) -- (U4) -- (V4);
                \node[fit=(f1)(f2), rounded corners=5pt, draw] (g) {};
                \node at ($(g.south) + (0,-.4)$) {$\bG$};
            \end{scope}

            \begin{scope}[xshift=38em]
             \node (e1) at (-2.2, .8) {};
                \node (e2) at (.4, -2.8) {};
                \node (f1) at (2.8, .8) {};
                \node (f2) at (1.6, -2.8) {};
                \node[cvertex] (x1) at (-1.8, -1.5) {};
                \node[cvertex] (x2) at (-2, -.5) {};
                \node[cvertex] (x3) at (-1.6, -1) {};
                \node[fit=(x1)(x2)(x3), rounded corners=3pt, draw=lipicsGray!40,
                    fill=lipicsGray!10] (z) {};
                \node at ($(z.north) + (0,.4)$) {$R$};
                \node[fit=(e1)(e2)(f1)(f2), rounded corners=5pt, draw] (f) {};
                \node at ($(f.south) + (0,-.4)$) {$F_{\bG}$};

                \node[cvertex] (x1) at (-1.7, -1.5) {};
                \node[cvertex] (x2) at (-2, -.5) {};
                \node[cvertex] (x3) at (-1.6, -.84) {};
                \node[cvertex] (x4) at (-2, -1.16) {};

                \node[vertex] (U3) at (2.6, -0) {};
                \node[vertex] (U4) at (1.1, -0) {};
                \node[vertex] (V2) at (2.6, -2) {};
                \node[vertex] (V3) at (1.85, -2) {};
                \node[vertex] (V4) at (1.1, -2) {};
                \draw[cedge] (U4) -- (x3) -- (V4);
                \draw[cedge] (U3) -- (x3) -- (V2);
                \draw[cedge] (x3) -- (V3);
                \draw[cedge] (U3) -- (x2);
                \draw[cedge] (U4) -- (x2);
                \draw[cedge] (V4) -- (x1) -- (x2) -- (x3) -- (x4);
                \draw[cedge] (V2) -- (x1);
                \draw[cedge] (V3) -- (x1);
                \draw[sdedge] (U3) -- (V4);
                \draw[sdedge] (U3) -- (V2);
                \draw[sdedge] (V3) -- (U4) -- (V4);
                \node[fit=(U3)(U4), rounded corners=3pt, draw=lipicsGray!40] (x1) {};
                \node[fit=(V2)(V3)(V4), rounded corners=3pt, draw=lipicsGray!40] (x2) {};
                \node at ($(x1.north) + (0,.3)$) {$V_1$};
                \node at ($(x2.south) + (0,-.3)$) {$V_2$};
         \end{scope}
        \end{tikzpicture}
        \caption{Replacing a pair of adjacent blocks $B_1$ and $B_2$ by a bipartite graph $\bG$.}\label{fig:implant}
    \end{figure}

It is worth to mention that Conjecture~\ref{conj:main} was previously confirmed for hereditary properties that only have a \emph{single forbidden induced subgraph}, i.e., for the case $|\obstr(\Phi)|=1$~\cite{RothSW20}. The corresponding proof uses the idea of replacing a single edge that we described in the initial construction. It then boils down to a reduction from counting independent sets. In our work, we significantly generalise the gadget construction, and by using $\#\bipindsubsprob(\Psi_\Phi)$, where $\Psi_\Phi$ depends on $\Phi$, we also broaden the class of problems we reduce from. Note that the problem of counting independent sets is a very special case of a property that distinguishes independent sets from bicliques.

We continue by giving an overview of the tight parameterised reduction from $\#\bipindsubsprob(\Psi_\Phi)$ to $\#\indsubsprob(\Phi)$.
Given a bipartite graph $\bG=(V_1, V_2, E)$ together with a  positive integer $k$ as input to $\#\bipindsubsprob(\Psi_\Phi)$ let $F_{\bG}$ be as defined previously. Let $R=V(F_{\bG})\setminus (V_1 \cup V_2)$, so $R$ contains the vertices of $H$ that are outside of $B_1$ and $B_2$ (which were replaced by $V_1$ and $V_2$ in the construction of $F_{\bG}$). Let $k'=k+|R|$. One can show that for the sought-for number of $k$-vertex induced subgraphs of $\bG$ that satisfy $\Psi_\Phi$ we have
\[
    \#\bipindsubs{\Psi_\Phi,k}{\bG} = \#\{S \in \indsubs{\Phi,k'}{F_{\bG}} \mid R \subseteq S\}.
\]
Then, from the standard inclusion-exclusion principle, it follows that 
\[
    \#\{S \in \indsubs{\Phi,k'}{F_{\bG}} \mid R \subseteq S\} = \sum_{J\subseteq R}{(-1)^{|J|} \cdot \indsubs{\Phi,k'}{F_{\bG}\setminus J}},
\]
where $F_{\bG}\setminus J$ is the graph obtained from $F_{\bG}$ by deleting the vertices in $J$.
Thus, an algorithm that makes $2^{|R|}\in O(1)$ oracle calls, each of the form $(F_{\bG}\setminus J, k')$, can compute the  value $\#\bipindsubs{\Psi_\Phi,k}{\bG}$, which gives the sought-for reduction, i.e., the connection between the bipartite property $\Psi_\Phi$ and the original graph property $\Phi$.

There is an additional ingredient which so far we have swept under the rug. In the definition of $F_{\bG}$ we replaced two adjacent blocks in the chosen graph $H$ by the graph $\bG$. It is important that the blocks $B_1$ and $B_2$ share an edge. (Since they are blocks, this means that there is a complete set of edges between them.) However, it is possible that $\Phi$ contains an independent set as forbidden induced subgraph. It would follow that the graph $H$ in $\obstr(\Phi)$ with edge-minimal $H\down$ is an independent set, which would spoil the construction. In this case it helps to consider a closely related property. For each graph $G$, let $\overline{G}$ be the complement of $G$. Then we define $\overline{\Phi}$ with $\overline{\Phi}(G)=\Phi(\overline{G})$. It is known that $\overline{\Phi}$ is hereditary if and only if $\Phi$ is hereditary. Furthermore, $\#\indsubsprob(\Phi)$ and $\#\indsubsprob(\overline{\Phi})$ are known to be tightly interreducible by parameterised reductions~\cite{RothSW20}. So, for all our purposes, we are free to work with either one of $\Phi$ or $\overline{\Phi}$. By a simple application of Ramsey's theorem, we show that every hereditary property $\Phi$, for which both $\Phi$ and $\overline{\Phi}$ have an independent set as forbidden induced subgraph, has to be \meagre.
Conversely, if $\Phi$ is not \meagre then at least one of $\Phi$ or $\overline{\Phi}$ is a suitable candidate for our construction.

The approach of utilising bipartite properties and implanting bipartite graphs into some fixed graph $H$ that depends on the graph property $\Phi$ is not only applicable to hereditary properties. With a slightly different construction we also prove Theorem~\ref{thm:homequivalence}, the classification for properties that are invariant under homomorphic equivalence. In this case it suffices to implant a graph $\bG$ into an edge $e$ in a graph $H$ for which $\Phi(H)\neq \Phi(H-e)$ holds. We omit further details but it is worth to point out that we actually classify a less natural but even more general class of properties. A graph property $\Phi$ is \emph{twin-invariant} if, for each pair of graphs $H_1$ and $H_2$ that have isomorphic twin-free quotients\footnote{In the definition we can even get away with only considering graphs whose twin-free quotient contains at least two vertices. This is a technicality which makes the class of covered properties more general and ensures that this result covers, for instance, also the property of being (dis)connected, which was of interest in some of the earlier works~\cite{JerrumM15, RothS18}.}, we have $\Phi(H_1)=\Phi(H_2)$.
This criterion covers previously unclassified properties such as ``disconnected or bipartite'' or ``disconnected or triangle-free'' but more importantly it is not hard to see that every property that is invariant under homomorphic equivalence is also twin-invariant. Thus, Theorem~\ref{thm:homequivalence} is a direct consequence of the following more general result.

\begin{restatable}{theorem}{twinequivalence}\label{thm:twinequivalence}
Let $\Phi$ be a computable twin-invariant graph property. If $\Phi$ is \meagre then $\#\indsubsprob(\Phi)$ is solvable in polynomial time. Otherwise, $\#\indsubsprob(\Phi)$ is $\#\W{1}$-complete and, assuming ETH, cannot be solved in time $f(k)\cdot |G|^{o(k)}$ for any function $f$.
\end{restatable}

\subsection{Further Related Work}\label{sec:furtherWork}
The decision version of $\#\indsubsprob(\Phi)$ for hereditary properties $\Phi$ was studied and fully classified by Khot and Raman~\cite{KhotR02} --- their results are summarised in Table~\ref{tab:results} --- and also by Eppstein, Gupta and Havvaei~\cite{EppsteinGH21} who additionally restricted the problem to hereditary classes of input graphs. Furthemore, if for each $k$, the property $\Phi$ is true for at most one $k$-vertex graph, the work of Chen, Thurley and Weyer~\cite{ChenTW08} establishes hardness for both, decision and counting, whenever $\Phi$ is not \meagre.\footnote{We remark that our notion of \meagre coincides with their notion of \meagre in the special case where $\Phi$ is true for at most one $k$-vertex graph for each $k$, which applies, e.g., to the properties of being a path, a cycle, or a matching.}

The complexity of computing an $\varepsilon$-approximation of $\#\indsubsprob(\Phi)$ was investigated by Jerrum and Meeks in a sequence of papers~\cite{JerrumM15,JerrumM15density,Meeks16,JerrumM17}, and in case of hereditary properties, it was ultimately resolved by Meeks~\cite{Meeks16} (see Table~\ref{tab:results}). For more general classes of properties, there are only partial results, and to the best of our knowledge the complexity of approximating $\#\indsubsprob(\Phi)$ is still open for edge-monotone properties. However, there are strong recent meta-theorems such as the $k$-Hypergraph framework due to Dell, Lapinskas, and Meeks~\cite{DellLM20} that yield efficient approximation algorithms for $\#\indsubsprob(\Phi)$ by reduction to vertex-coloured decision problems.

Most results on $\#\indsubsprob(\Phi)$ are concerned with hardness of exact counting; we list them chronologically.
\begin{itemize}
    \item In~\cite{JerrumM15}, Jerrum and Meeks proved that $\#\indsubsprob(\Phi)$ always belongs to $\#\W{1}$, given that $\Phi$ is computable. Additionally, $\#\W{1}$-hardness was established for the property $\Phi$ of being connected.
    \item In~\cite{JerrumM15density}, $\#\W{1}$-hardness was proved by the same authors for all properties with low edge-densities. This covers, for example, all non-trivial minor-closed properties.
    \item In the survey paper of Meeks~\cite{Meeks16} a $\#\W{1}$-hardness result was established for properties that are closed under the addition of edges, and whose edge-minimal elements have unbounded treewidth.
    \item In~\cite{JerrumM17}, Jerrum and Meeks proved $\#\W{1}$-hardness for the property of having an even, or an odd, number of edges.
    \item In the breakthrough paper of Curticapean, Dell and Marx~\cite{CurticapeanDM17}, the principle of complexity monotonicity was introduced, and it was shown that $\#\indsubsprob(\Phi)$ is always either fixed-parameter tractable or $\#\W{1}$-hard, given that $\Phi$ is computable. (See Theorem~\ref{thm:intro_monotonicity}.)
    \item In~\cite{RothS18}, Roth and Schmitt established $\#\W{1}$-hardness and a tight lower bound under ETH for edge-monotone properties that are non-\meagre and false on odd cycles.
    \item In~\cite{DorflerRSW22}, D\"orfler, Roth, Schmitt, and Wellnitz introduced the ``algebraic approach to hardness'' and established $\#\W{1}$-hardness for properties that distinguish independent sets from so-called wreath graphs. As a concrete example, their result applies to monotone properties that are non-trivial on bipartite graphs, in which case a tight lower bound under ETH is also achieved.
    \item Finally, in~\cite{RothSW20}, Roth, Schmitt and Wellnitz established $\#\W{1}$-hardness and an \emph{almost} tight conditional lower bound of the form $f(k)\cdot n^{o(k/\sqrt{\log k})}$ for non-\meagre monotone properties. More generally, they proved the result for any property with so-called $f$-vectors of small Hamming weight; we refer the reader to~\cite[Sections 3 and 4]{RothSW20} for a detailed exposition but remark that hereditary properties do not, in general, have $f$-vectors with small enough Hamming weight for the result in the current paper to be covered by their meta theorem.
    
    Additionally, Roth, Schmitt and Wellnitz proved Conjecture~\ref{conj:main} for the restricted case of hereditary properties that are defined by a single forbidden induced subgraph.
\end{itemize}

None of the previous results comes close to resolving Conjecture~\ref{conj:main} for all hereditary properties. In particular, since every monotone property is also hereditary, we not only subsume the classification for monotone properties from~\cite{RothSW20}, but we also improve the conditional lower bound from almost tight to tight.

\subsection{Open Problems}\label{sec:conclusions}
We conclude our presentation with two open problems and suggestions for further work. 

First, the most important open question is whether Conjecture~\ref{conj:main} is indeed true for all computable properties. With the case of hereditary properties now being resolved, the other central remaining family of properties with a natural closure condition is the class of all edge-monotone properties. Between the results from~\cite{Meeks16} and~\cite{RothS18}, large classes of edge-monotone properties are already covered, but there remains a significant gap towards a complete understanding. For example, none of the existing partial results resolves the complexity of $\#\indsubsprob(\Phi)$ for the following (slightly artificial) edge-monotone property:
\begin{center}
    $\Phi(H)=1$ if and only if $H$ is bipartite or has no apex.\footnote{An apex is a vertex adjacent to all other vertices.}
\end{center}
For this reason, we suggest to tackle Conjecture~\ref{conj:main} for the case of edge-monotone properties as a concrete next step. Our hope is that the technical framework that we introduce in the work at hand will help to close the remaining gap for edge-monotone properties as well.

However, we wish to point out that, in this case, strengthening the intractability part of Conjecture~\ref{conj:main} by additionally aiming for tight conditional lower bounds might require hardness assumptions stronger than ETH. The reason for the latter is the existence of artificial ``non-dense'' edge-monotone properties, such as
\[\Phi(H) = \begin{cases} 1 & \exists \ell: H = I_{a(\ell,\ell)}  \\ 0 & \text{otherwise}\,, \end{cases} \]
where $a$ is the Ackermann function. Observe that $\Phi$ is closed under the removal of edges (but not under the removal of vertices). It is easy to establish $\#\W{1}$-hardness of $\#\indsubsprob(\Phi)$ by reducing from the parameterised problem of counting independent sets using standard methods. However, the parameter explodes drastically in this reduction due to the fact that we can only reduce to instances of $\#\indsubsprob(\Phi)$ in which $k$ is in the image of the Ackermann function. This prevents us from coming even close to a tight ETH-based lower bound. One way to circumvent this problem is to restrict ourselves to graph properties that are \emph{dense} in the sense that the set of $\ell$ for which $\Phi$ is non-trivial on $\ell$-vertex graphs is a dense enough subset of the natural numbers; this approach was formalised and used in previous work~\cite{RothS18,DorflerRSW22,RothSW20}.

Finally, given that this work provides even more evidence for the intractability of $\#\indsubsprob(\Phi)$, we stress that a relaxation of the problem is unavoidable if efficient algorithms are sought. The obvious and most promising candidate for such a relaxation is to only aim for an approximation of the solution. 

As summarised in Table~\ref{tab:results}, Meeks~\cite{Meeks16} explicitly and exhaustively identified those hereditary properties $\Phi$ for which $\#\indsubsprob(\Phi)$ admits an FPTRAS. In particular, our main result shows that Meeks' result cannot be strengthened to yield fixed-parameter tractability of exact counting, unless ETH fails. 

However, for general (not necessarily hereditary) properties much less is known about the complexity of approximating $\#\indsubsprob(\Phi)$; again, we refer the reader to the survey of Meeks for a comprehensive overview~\cite{Meeks16}.

\subsection*{Acknowledgements}
We are very grateful to Johannes Schmitt and Philip Wellnitz for helpful comments on early and recent drafts of this work.

\newpage

\section{Preliminaries}\label{sec:prelim}\label{sec:background} 
Given a finite set $S$, we write $|S|$, or $\#S$ for its cardinality. Graphs in this work are undirected and without self-loops. We assume further that graphs are encoded by their adjacency matrix.\footnote{As we aim to establish lower bounds ruling out algorithms with running time $f(k)\cdot |G|^{o(k)}$ any reasonable encoding of graphs yields the same lower bounds.} Given a graph $G$, we write $\overline{G}$ for its complement, that is, $V(\overline{G})=V(G)$ and for each pair of distinct vertices $u,v\in V(G)$ we have $\{u,v\}\in E(\overline{G})$ if and only if $\{u,v\}\notin E(G)$.

\paragraph*{Homomorphisms, Induced Subgraphs, and Coloured Variants}
A \emph{homomorphism}  from a graph $G$ to a graph $H$ is an edge-preserving mapping $\varphi$ from $V(G)$ to $V(H)$, that is, for every edge $\{u,v\}\in E(G)$ we have $\{\varphi(u),\varphi(v)\}\in E(H)$. Such a homomorphism is also referred to as an $H$-colouring of $G$. Injective homomorphism are also known as \emph{embeddings}. A \emph{strong embedding} is an embedding $\varphi$ for which  $\{u,v\}\in E(G)$ if and only if $\{\varphi(u),\varphi(v)\}\in E(H)$. A bijective strong embedding is an \emph{isomorphism}, and an isomorphism from a graph $G$ to itself is an \emph{automorphism} of $G$.

Given a graph $G$ with an $H$-colouring $c$, we say that a homomorphism $\varphi$ from $H$ to $G$ is \emph{$c$-colour-prescribed} if, for each $v\in V(H)$, we have $c(\varphi(v))=v$, that is, $\varphi$ maps a vertex $v\in V(H)$ to a vertex in $G$ which is coloured with $v$. If the choice of $c$ is not specified or clear from context we just say that $\varphi$ is colour-prescribed. Note that a $c$-colour-prescribed homomorphism can only exist if $c$ is surjective. Therefore, we will generally only be interested in $H$-colourings $c$ that are surjective, as otherwise the computational problems we consider become trivial.

Given a subset $A$ of the edges of a graph $H$, we write $H[A]\coloneqq (V(H),A)$ for the edge-subgraph of $H$ with respect to $A$; note that $H[A]$ might contain isolated vertices. Note that we slightly overload notation here: For a subset of vertices $S\subseteq V(H)$, the expression $H[S]$ denotes the induced subgraph $(S,E')$ of $H$, where $E'\subseteq E(H)$ contains all edges of $H$ between vertices in $S$.

The notion of colour-prescribed homomorphisms readily extends to edge-subgraphs, and we will mostly use it in this context: Given a graph $H$, a subset $A\subseteq E(H)$, and an $H$-colouring $c$ of a graph $G$, a homomorphism $\varphi$ from $H[A]$ to $G$ is \emph{($c$-)colour-prescribed} if, for each $v\in V(H)$, we have $c(\varphi(v))=v$.

We write $\cphoms{H[A]}{c}{G}$ for the set of all $c$-colour-prescribed homomorphisms from $H[A]$ to $G$\footnote{We decided to make the $H$-colouring $c$ of $G$ explicit in our notation of colour-prescribed homomorphisms by using $\to_c$. Let us point out, however, that $\cphoms{H[A]}{c}{G}$ is the same as $\cphoms{H[A]}{H}{G}$ in~\cite{DorflerRSW22}, where the colouring $c$ was assumed to be given implicitly.}.
Following the convention of~\cite{DorflerRSW22}, it will be convenient to write $\cpindsubs{H[A]}{c}{G}$ for the set of all images of $c$-colour-prescribed strong embeddings from $H[A]$ to $G$. Formally, $\cpindsubs{H[A]}{c}{G}$ contains all $S\subseteq V(G)$ with $|S|=|V(H)|$, $c(S)=V(H)$, and $\{u,v\}\in E(G[S])$ if and only if $\{c(u),c(v)\}\in A$. In other words, it contains all subsets $S$ of vertices of $G$ for which $c\vert_S$ is an isomorphism from $G[S]$ to $H[A]$.

\paragraph*{False Twins and Twin-free Quotients}
Two vertices $u$ and $v$ of a graph $H$ are called \emph{false twins} if they have the same neighbourhood; note that the latter implies that $u$ and $v$ are not adjacent since we do not allow self-loops.

The notion of false twins induces an equivalence relation on $V(H)$ in the canonical way: Two vertices are set to be equivalent if and only if they are false twins. The equivalence classes are called \emph{blocks}, and we define the \emph{twin-free quotient} $H\down$ of $H$ as the quotient graph w.r.t.\ these blocks, that is, $H\down$ contains a vertex for each block, and two blocks $B_1$ and $B_2$ are adjacent if and only if there are vertices $u\in B_1$ and $v\in B_2$ such that $\{u,v\}\in E(H)$. Note that, by definition of false twins, the latter is in fact equivalent to $H$ containing a complete set of edges between the vertices in $B_1$ and the vertices in $B_2$.

We will be using the fact that the twin-free quotient of an induced subgraph can never have more edges than the twin-free quotient of the entire graph.
\begin{lemma}\label{lem:X}
    Let $H$ be a graph and let $S$ be a subset of its vertices. Then $H[S]\down$ has at most as many edges as $H\down$.
\end{lemma}
\begin{proof}
    Since every pair of vertices in $S$ that are false twins in $H$ are also false twins in $H[S]$, it follows that $H[S]\down$ is isomorphic to $(H\down[S])\down$
    \footnote{Technically, we assume that each vertex of the quotient $H\down$ is obtained by identifying all vertices of the corresponding block, and thereby keeps all the corresponding labels. (So that $S$ is a subset of the vertices of $H\down$, where some vertices in $S$ might also refer to the same vertex of $H\down$.)}.
    So we have
    \[
        \#E(H[S]\down)=\#E((H\down[S])\down)\le \#E(H\down[S])\le \#E(H\down).
    \]
\end{proof}

\paragraph{Graph Properties}
A graph property $\Phi$ is a function from the set of graphs to $\{0,1\}$ that is invariant under graph isomorphisms, i.e., if $H_1$ and $H_2$ are isomorphic then $\Phi(H_1)=\Phi(H_2)$. A graph $H$ \emph{satisfies} $\Phi$ if $\Phi(H)=1$. 
We write $\overline{\Phi}$ for the \emph{inverse}\footnote{We avoid confusion and use the term ``inverse'' rather than ``complement'', which has different meanings in the literature.} of $\Phi$, which is defined as follows:
\[\overline{\Phi}(H)=1 \Leftrightarrow \Phi(\overline{H})=1\,,\]
where $\overline{H}$ is the previously defined complement of $H$.
A graph property $\Phi$ is \emph{$k$-trivial} for some positive integer $k$ if it is constant ($0$ or $1$) on all $k$-vertex graphs. The property $\Phi$ is \emph{\meagre}\footnote{Our notion of meagre graph properties generalises the notion of \emph{meagre graph classes} as introduced by Chen, Thurley and Weyer~\cite{ChenTW08}: If, for each $k$, $\Phi$ is true for at most one $k$-vertex graph, then $\#\indsubsprob(\Phi)$ is equivalent to the problem $\#\textsc{StrEmb}(\mathcal{C})$ as studied in~\cite{ChenTW08}, where $\mathcal{C}$ is the class of all graphs that satisfy $\Phi$. In this case $\Phi$ is meagre if and only if $\mathcal{C}$ is.} if there are only finitely many $k$ for which $\Phi$ is \emph{not} $k$-trivial. In other words $\Phi$ is \meagre if there are only finitely many $k$ for which there are $k$-vertex graphs $H_1$ and $H_2$ such that $\Phi(H_1)\neq\Phi(H_2)$. In particular, there is a constant $B$ such that $\Phi$ is $k$-trivial for all $k\geq B$.

A graph property $\Phi$ is \emph{hereditary} if it is closed under taking induced subgraphs, that is, for each graph $H$ and $S\subseteq V(H)$, $\Phi(H)=1$ implies $\Phi(H[S])=1$. Note that hereditary properties are precisely those properties that are closed under vertex-deletion. Each hereditary property $\Phi$ is defined by a (possibly infinite) set $\obstr(\Phi)$ of forbidden induced subgraphs such that $\Phi(H)=1$ if and only if no induced subgraph of $H$ is contained in $\obstr(\Phi)$. The following observation will be convenient in our main reduction. Note that the implication from item~\ref{item:obs2} to item~\ref{item:obs3} follows from a simple application of Ramsey's theorem, as mentioned in the footnote in Table~\ref{tab:results}.
\begin{obs}\label{obs:hereditarynontriv}
For a hereditary property $\Phi$ the following are equivalent:
\begin{enumerate}
    \item $\Phi$ is not meagre.
    \item $\Phi$ is true for infinitely many graphs and $\obstr(\Phi)\neq \emptyset$.\label{item:obs2}
    \item $\Phi$ is not constant true, and there are no $s$ and $t$ such that $\Phi(K_s)=\Phi(I_t)=0$\label{item:obs3} (cf.~Table~\ref{tab:results})
\end{enumerate}
\end{obs}

Given a graph property $\Phi$, a positive integer $k$, and a graph $G$, we write $\indsubs{\Phi,k}{G}$ for the set of $k$-vertex induced subgraphs of $G$ that satisfy $\Phi$. More formally
\[\indsubs{\Phi,k}{G} = \{S \subseteq V(G)\mid |S|=k ~\wedge~ \Phi(G[S])=1\}\,.  \]

\paragraph*{Bipartite Properties}
For what follows, we assume that bipartite graphs come with an (ordered) bipartition: Formally, a \emph{bipartite graph} is a triple $\bG=(V_1,V_2,E)$, where $V_1$ and $V_2$ are (possibly empty) vertex sets, and $E$ is a set of (undirected) edges between $V_1$ and $V_2$. We set $V(\bG)\coloneqq V_1\cup V_2$, $V_1(\bG)\coloneqq V_1$, and $V_2(\bG)\coloneqq V_2$. 
The \emph{underlying graph} of $\bG$ is the graph $G\coloneqq(V_1 \cup V_2, E)$. Note that we use different typestyles for bipartite graphs (which have a specified ordered bipartition) and the corresponding underlying graph, for which no particular bipartition is specified.

A homomorphism $\varphi$ from a bipartite graph $\bG=(V_1,V_2,E)$ to a bipartite graph $\hat{\bG}=(\hat{V}_1,\hat{V}_2,\hat{E})$ is a homomorphism from the underlying graph of $\bG$ to the underlying graph of $\hat{\bG}$. Such a homomorphism is a \emph{consistent} homomorphism from $\bG$ to $\hat{\bG}$ if $\varphi(V_i)\subseteq \hat{V}_i$ for $i\in\{1,2\}$. If $\varphi$ is an isomorphism then $\bG$ and $\hat{\bG}$ are \emph{consistently isomorphic}.
Similarly as for (not necessarily bipartite) graphs, given bipartite graphs $\bH$ and $\bG$, a consistent $\bH$\emph{-colouring} of $\bG$ is a consistent homomorphism from $\bG$ to $\bH$.
\begin{obs}
Let $\bH$ and $\bG$ be bipartite graphs with corresponding underlying graphs $H$ and $G$, respectively.
If $c$ is a consistent $\bH$-colouring of $\bG$, then $c$ is also an $H$-colouring of $G$. 
\end{obs}

The $(k,k)$\emph{-biclique} is the bipartite graph $\bB_{k,k}\coloneqq(L,R,E)$ where $|L|=|R|=k$ and $E$ contains all edges between $L$ and $R$. 

Let $\bG=(V_1,V_2,E)$ be a bipartite graph with a subset of vertices $S\subseteq V_1\cup V_2$. Let $E'\subseteq E$ contain all edges in $E$ between $V_1\cap S$ and $V_2\cap S$. Then $\bG[S]=(V_1\cap S, V_2\cap S, E')$ is an \emph{induced subgraph} of $\bG$.
Given a subset $A$ of the edges of a bipartite graph $\bH$, we write $\bH[A]\coloneqq (V_1(\bH),V_2(\bH),A)$ for the edge-subgraph of $\bH$ with respect to $A$; note that $\bH[A]$ might contain isolated vertices, and that the underlying graph of $\bH[A]$ is $H[A]$.

A \emph{bipartite property} $\Psi$ is a function from the set of bipartite graphs to $\{0,1\}$ that is invariant under \emph{consistent} isomorphisms, i.e., if $\bH_1$ and $\bH_2$ are consistently isomorphic then $\Psi(\bH_1)=\Psi(\hat{\bH_2})$. 
Every graph property $\Phi$ can be translated into a bipartite property by mapping every bipartite graph to $1$ for which the underlying graph satisfies $\Phi$. However, there are bipartite properties that cannot be translated into graph properties. For example, consider $\Psi((V_1,V_2,E))=1$ if and only if $|V_1|=1$. For $V_1=\{c\}$, $V_2=\{l,r\}$, and $E=\{\{l,c\},\{c,r\}\})$, the bipartite graph $(V_1, V_2, E)$ is in $\Psi$, whereas $(V_2, V_1, E)$ is not, despite the fact that the corresponding underlying graphs are isomorphic.

Given a bipartite graph $\bG=(V_1,V_2,E)$, a positive integer $k$, and a bipartite property $\Psi$, we define
\[\bipindsubs{\Psi,k}{\bG} = \{S \subseteq V_1\cup V_2 \mid |S|=k ~\wedge~ \Psi(\bG[S])=1\}\,.  \]
Given a bipartite graph $\bG=(V_1,V_2,E)$ with a \emph{consistent} $\bH$-colouring $c$, $\cpbipindsubs{\Psi}{c}{\bG}$
is the set of all $S\subseteq V_1\cup V_2$ of size $|V(\bH)|$ such that $c(S)=V(\bH)$ (that is, $S$ must include precisely one vertex per colour) and $\Psi(\bG[S])=1$. We will refer to the elements of $\cpbipindsubs{\Psi}{c}{\bG}$ as \emph{($c$-)colour-prescribed} induced subgraphs of $\bG$ that satisfy $\Psi$.

\paragraph{A Colourful Bipartite Tensor Product}
Given two graphs $G, \hat{G}$ with corresponding $H$-colourings $c$ and $\hat{c}$, the \emph{colour-prescribed tensor product} $(G,c) \timesc (\hat{G},\hat{c})$ is the graph with vertex set $\{(v,\hat{v})\in V(G)\times V(\hat{G}) \mid c(v)=\hat{c}(\hat{v})\}$ and edges between $(u,\hat{u})$ and $(v,\hat{v})$ whenever $\{u,v\}$ is an edge in $G$ and $\{\hat{u},\hat{v}\}$ is an edge in $\hat{G}$.\footnote{Again we make the colourings explicit in this notation. $(G,c) \timesc (\hat{G},\hat{c})$ coincides with the graph $G\times_H \hat{G}$ as introduced in~\cite{DorflerRSW22}.}

We will need a bipartite version of this tensor product.
Given two bipartite graphs $\bG$ and $\hat{\bG}$ with consistent $\bH$-colourings $c$ and $\hat{c}$, respectively, the \emph{bipartite colour-prescribed tensor product}, denoted by $(\bG, c)\timescc (\hat{\bG},\hat{c})$, is a bipartite graph with bipartition $(V_1, V_2)$ defined as follows:
\begin{itemize}
    \item $V_1\coloneqq \{ (v,\hat{v})\in V_1(\bG)\times V_1(\hat{\bG}) \mid c(v)=\hat{c}(\hat{v}) \}$
    \item $V_2\coloneqq \{ (v,\hat{v})\in V_2(\bG)\times V_2(\hat{\bG}) \mid c(v)=\hat{c}(\hat{v}) \}$
    \item $(u,\hat{u})$ and $(v,\hat{v})$ are adjacent if and only if $\{u,v\}\in E(\bG)$ and $\{\hat{u},\hat{v}\}\in E(\hat{\bG})$.
\end{itemize}
Observe that $c$ and $\hat{c}$ correspond to  consistent $\bH$-colourings of $(\bG, c)\timescc(\hat{\bG},\hat{c})$ by mapping $(v,\hat{v})\mapsto c(v) (=\hat{c}(\hat{v}))$.
Observe further, that $V_1\cup V_2= \{ (v,\hat{v})\in V(G)\times V(\hat{G}) \mid c(v)=\hat{c}(\hat{v})\}$, where $G$ and $\hat{G}$ are the underlying graphs of $\bG$ and $\hat{\bG}$, respectively. The underlying graph of $(\bG, c)\timescc (\hat{\bG}, \hat{c})$ is, in fact, the colour-prescribed tensor product $(G, c)\timesc (\hat{G},\hat{c})$ --- a corresponding $H$-colouring is given by $(v,\hat{v})\mapsto c(v) (=\hat{c}(\hat{v}))$. We state the aforementioned fact formally:

\begin{fact}\label{fact:convenient}
The underlying graph of $(\bG, c)\timescc (\hat{\bG}, \hat{c})$ is $(G, c)\timesc (\hat{G},\hat{c})$, where $c$ is an $H$-colouring of $(G, c)\timesc (\hat{G},\hat{c})$.
\end{fact}

\paragraph{Parameterised and Fine-Grained Complexity}
A detailed introduction to the field of parameterised complexity can be obtained from one of several textbooks on the subject, such as~\cite{CyganFKLMPPS15, FlumG06}, where~\cite{FlumG06} also explicitly covers parameterised counting problems.
A \emph{parameterised counting problem} is a function $P:\{0,1\}^\ast \rightarrow \mathbb{N}$ together with a computable parameterisation $\kappa: \{0,1\}^\ast \rightarrow \mathbb{N}$.
We will focus on the subsequent parameterised counting problems; in what follows, $\mathcal{H}$ denotes a class of graphs, $\Phi$ denotes a graph property, and $\Psi$ denotes a bipartite property.

\begin{itemize}[leftmargin = 3.75cm]
\item[$\#\clique$] \problem{A graph $G$ and a positive integer $k$}{$\kappa(G,k)= k$}{The number of cliques of size $k$ in $G$}
\item[$\#\indsubsprob(\Phi)$] \problem{A graph $G$ and a positive integer $k$}{$\kappa(G,k)= k$}{$\#\indsubs{\Phi,k}{G}$, i.e., the number of $k$-vertex induced subgraphs of $G$ that satisfy $\Phi$.}
\item[$\#\bipindsubsprob(\Psi)$] \problem{A bipartite graph $\bG$ and a positive integer $k$}{$\kappa(\bG,k)= k$}{$\#\bipindsubs{\Psi,k}{\bG}$, i.e., the number of $k$-vertex induced bipartite subgraphs of $\bG$ that satisfy $\Psi$}\footnote{We stress that $\#\bipindsubsprob(\Psi)$ is not equivalent to the restriction of $\#\indsubsprob(\Phi)$ to bipartite input graphs since the bipartitions of the input graph and all of its induced subgraphs are fixed in case of $\#\bipindsubsprob(\Psi)$. Also, on bipartite graphs, bipartite properties are more expressive than graph properties.}
\item[$\#\cphomsprob(\mathcal{H})$] \problem{A graph $H\in \mathcal{H}$, a graph $G$, and an $H$-colouring $c$ of $G$}{$\kappa(H,G,c)= |V(H)|$}{$\#\cphoms{H}{c}{G}$, i.e., the number of $c$-colour-prescribed homomorphisms from $H$ to $G$}
\item[$\#\cpbipindsubsprob(\Psi)$] \problem{Bipartite graphs $\bH$ and $\bG$, and a consistent $\bH$-colouring $c$ of $\bG$}{$\kappa(\bH,\bG,c)= |V(\bH)|$}{$\#\cpbipindsubs{\Psi}{c}{\bG}$, i.e., the number of $c$-colour-prescribed induced subgraphs of $\bG$ that satisfy $\Psi$}
\end{itemize}

A parameterised counting problem $(P,\kappa)$ is \emph{fixed-parameter tractable (FPT)} if there is a computable function $f$ and an algorithm $\mathbb{A}$, such that $\mathbb{A}$ computes $P(x)$ in time $f(\kappa(x))\cdot |x|^{O(1)}$; $\mathbb{A}$ is referred to as \emph{FPT algorithm}. For example, the problem $\#\indsubsprob(\Phi)$ is FPT whenever $\Phi$ is computable and \meagre; the following result is simple and well-known, and implicit in~\cite{RothSW20}, however, we provide a proof for completeness:

\begin{lemma}\label{lem:meagre_easy}
Let $\Phi$ be a computable \meagre graph property. Then $\#\indsubsprob(\Phi)$ is FPT.
\end{lemma}
\begin{proof}
Since $\Phi$ is \meagre, there exists a constant $B$ such that for all $k\geq B$ we have that $\Phi$ is either constant true or constant false on the set of all $k$-vertex graphs. Consequently, $\#\indsubsprob(\Phi)$ can be solved as follows: On input $(G,k)$, we check whether $k<B$. If so, we solve the problem by brute force, i.e., we iterate over all $k$-vertex subsets of $G$ and count how many of them induce a subgraph that satisfies $\Phi$. Otherwise, the output is trivially $0$ or $\binom{|V(G)|}{k}$, depending on whether $\Phi$ is constant true or constant false on the set of all $k$-vertex graphs. Since $\Phi$ is computable, the latter can be decided (in time only depending on $k$). Since $B\in O(1)$, the overall running time is bounded by $f(k)\cdot |x|^{O(1)}$ for some computable function $f$.
\end{proof}
Note that we cannot improve fixed-parameter tractability to polynomial-time solvability in the generality of the previous lemma, since there are \meagre computable graph properties for which $\#\indsubsprob(\Phi)$ is not solvable in polynomial time. For example, given a $2$-$\mathrm{EXP}$-complete problem $\mathcal{L}\subseteq \{0,1\}^\ast$, we can set $\Phi(H)=1$ if and only if $\mathsf{bin}(|V(H)|) \in \mathcal{L}$.

A \emph{parameterised Turing-Reduction} from $(P,\kappa)$ to $(\hat{P},\hat{\kappa})$ is an FPT algorithm for $(P,\kappa)$ which is equipped with oracle access to $\hat{P}$ and additionally satisfies that there is a computable function $g$ such that on input $x$, the parameter $\hat{\kappa}(y)$ of each oracle query $y$ is bounded by $g(\kappa(x))$. We write $(P,\kappa)\fptred (\hat{P},\hat{\kappa})$ if such a parameterised Turing-reduction exists.

A parameterised counting problem $(P,\kappa)$ is $\#\W{1}$-\emph{hard} if $\#\clique \fptred (P,\kappa)$. Evidence for the (fixed-parameter) intractability of $\#\W{1}$-hard problems is given by the Exponential Time Hypothesis:

\begin{conjecture}[ETH~\cite{ImpagliazzoP01}]
The \emph{Exponential Time Hypothesis} (ETH) asserts that the problem $3\textsc{-SAT}$ cannot be solved in time $\exp(o(n))$ where $n$ is the number of variables of the input formula.
\end{conjecture}

\cite{Chenetal05,Chenetal06} proved that $\#\clique$ cannot be solved in time $ f(k) \cdot |G|^{o(k)}$ for any function $f$, unless ETH fails.\footnote{In fact,~\cite{Chenetal05,Chenetal06} proved that \emph{detecting} a clique of size $k$ cannot be done in time $f(k)\cdot |V(G)|^{o(k)}$, unless ETH fails. Our statement follows by observing that the number of $k$-cliques certainly reveals whether one exists, and using that $|V(G)|^{o(k)}=|G|^{o(k)}$.} In particular, the latter implies that $\#\W{1}$-hard problems are not fixed-parameter tractable, unless ETH fails. However, since we aim for matching lower bounds under ETH, rather than just $\#\W{1}$-hardness, we introduce a more restricted notion of reductions.

\begin{defn}
A parameterised Turing-Reduction from $(P,\kappa)$ to $(\hat{P},\hat{\kappa})$ is called \emph{tight} if, on input $x$, we have that $\hat{\kappa}(y) \in O(\kappa(x))$ for each oracle query $y$. We write $(P,\kappa)\tightred(\hat{P},\hat{\kappa})$ if a tight parameterised Turing-Reduction exists.
\end{defn}

The following simple lemma will allow us to derive our desired lower bounds using the above notion of tight reductions. 

\begin{lemma}\label{lem:tightred}
Let $(P,\kappa)$ and $(\hat{P},\hat{\kappa})$ be parameterised counting problems such that
\begin{itemize}
    \item $(P,\kappa)$ cannot be solved in time $f(\kappa(x))\cdot |x|^{o(\kappa(x))}$ for any function $f$, and
    \item $(P,\kappa)\tightred (\hat{P},\hat{\kappa})$.
\end{itemize}
Then $(\hat{P},\hat{\kappa})$ cannot be solved in time $f(\hat{\kappa}(x))\cdot |x|^{o(\hat{\kappa}(x))}$ for any function $f$.
\end{lemma}
\begin{proof}
Let $\mathbb{A}$ be the (tight) parameterised Turing-reduction from $(P,\kappa)$ to $(\hat{P},\hat{\kappa})$, that is, there exists a (computable) function $g$ such that $\mathbb{A}$ computes, on input $x$, the value $P(x)$ in time $g(\kappa(x))\cdot |x|^{O(1)}$. Additionally, $\mathbb{A}$ has oracle access to $\hat{P}$ and $\hat{\kappa}(y)\in O(\kappa(x))$ for each oracle query $y$. 

Let us assume for contradiction that there exists a function $\hat{f}$ and an algorithm $\hat{\mathbb{A}}$ such that, on input $y$, $\hat{\mathbb{A}}$ computes $\hat{P}(y)$ in time $\hat{f}(\hat{\kappa}(y))\cdot |y|^{o(\hat{\kappa}(y))}$. An algorithm for $(P,\kappa)$ is obtained by running $\mathbb{A}$ and by simulating each oracle query by running $\hat{\mathbb{A}}$. Let us bound the overall running time on input $x$:

First, the size of each oracle query $y$ is bounded by $g(\kappa(x))\cdot |x|^{O(1)}$, since this is the running time bound of $\mathbb{A}$. Further, we have $\hat{\kappa}(y)\in O(\kappa(x))$. Consequently, each oracle query can be simulated in time
 \[\hat{f}(\hat{\kappa}(y))\cdot |y|^{o(\hat{\kappa}(y))} \leq \hat{f}(O(\kappa(x))) \cdot \left(g(\kappa(x))\cdot |x|^{O(1)} \right)^{o(\kappa(x))}  = \hat{f}(O(\kappa(x))) \cdot g(\kappa(x))^{o(\kappa(x))} \cdot |x|^{o(\kappa(x))}\,.\]
 The total running time can thus be (generously) bounded by
 \[g(\kappa(x))\cdot |x|^{O(1)} \cdot \hat{f}(O(\kappa(x))) \cdot g(\kappa(x))^{o(\kappa(x))} \cdot |x|^{o(\kappa(x))} \leq f(\kappa(x)) \cdot |x|^{o(\kappa(x))} \,, \]
 where $f$ is any function with $f(k)\geq g(k)\cdot \hat{f}(O(k)) \cdot g(k)^{o(k)} $. This yields the desired contradiction and thus concludes the proof.
\end{proof}

\paragraph*{A note on $\#\W{1}$-completeness:}
While, intuitively, the class $\#\W{1}$ contains all parameterised counting problems that are reducible to $\#\textsc{Clique}$, the formal definition is rather technical, and we refer the interested reader to the excellent standard textbook of Flum and Grohe~\cite[Chapter 14]{FlumG06} for a detailed exposition. Since $\#\indsubsprob(\Phi)$ was shown to be contained in $\#\W{1}$ for every computable property $\Phi$ by Jerrum and Meeks~\cite{JerrumM15}, we will only argue about $\#\W{1}$-hardness in the current paper.

\paragraph{Group Theory}

Given a group $(\Gamma,\ast)$ with identity $e$ and a set $X$, a \emph{group action} is a function $f$ from $\Gamma\times X$ to $X$ such that, for each $x\in X$ and $g_1, g_2 \in \Gamma$, it holds that $f(e,x)= x$ and $f(g_1, f(g_2,x))= f(g_1\ast g_2, x)$.
If the group action is clear from context we use the shortened notation $g\triangleright x$ instead of $f(g,x)$. For each $x\in X$, $\orb_{\Gamma}(x) =\{g\triangleright x \mid g\in \Gamma\}$ is the \emph{orbit} of $x$, and if $\Gamma$ is clear from context we drop the subscript. If $\orb(x)=\{x\}$ then $x$ is a \emph{fixed point}.
A group action is \emph{transitive} on $X$ if there is only a single orbit.
The \emph{stabiliser} of $x\in X$ by $\Gamma$ is the set $\stab_{\Gamma}(x)=\{g\in \Gamma \mid g \triangleright x= x\}$, again we drop the subscript if it is clear. The Orbit-Stabiliser theorem states that, for each $x\in X$, it holds that $|\Gamma|=|\orb(x)| \cdot |\stab(x)|$.

A finite group $\Gamma$ is a \emph{$p$-group} if its order $|\Gamma|$ is a power of $p$.  A \emph{$p$-Sylow subgroup} of a group $\Gamma$ is a maximal $p$-subgroup of $\Gamma$, i.e., it is not a proper subgroup of any other $p$-subgroup of $\Gamma$. The next theorem is part of Sylow's theorems.
\begin{theorem}[Sylow]\label{thm:Sylow}
    Let $\Gamma$ be a finite group and let $p$ be a prime with multiplicity $\ell\ge 1$ in the factorisation of the order $|\Gamma|$. Then there exists a $p$-Sylow subgroup of $\Gamma$ of order $p^\ell$.
\end{theorem}
The following result is a well-known application of Sylow's theorems and a detailed proof is given for instance in~\cite[Lemma 2]{DorflerRSW22}.
\begin{lemma}\label{lem:Sylow}
    Let $\Gamma$ be a finite group acting transitively on  a set $X$ such that $|X|=p^\ell$ for some $\ell\ge 0$. Then the induced action of any $p$-Sylow subgroup $\Gamma'\subseteq \Gamma$ on $X$ is also transitive.
\end{lemma}

\section{Bipartite Properties}\label{sec:bipartite}

The goal of this section is to establish Theorem~\ref{thm:main_bipartite}, which states a hardness result for bipartite properties. The result is similar to the one presented in~\cite[Theorem 1]{DorflerRSW22} and we follow their approach but adapt it to bipartite properties.
The crucial piece is Lemma~\ref{lem:bipartite_biclique_coeff} where we adapt the coefficient analysis. The subtle difference to~\cite{DorflerRSW22} is that for bipartite properties we cannot simply use~\cite[Lemma 1]{DorflerRSW22} (which is the heart of their algebraic approach). The issue is that a bipartite property is not necessarily invariant under all automorphisms. It is only known to be invariant under those automorphisms that respect the given bipartitions, i.e., under consistent automorphisms. The new coefficient analysis then relies on the fact that even consistent automorphisms act transitively on the edges of a $(k,k)$-biclique. 

As in several previous works, we use a technique by Curticapean, Dell and Marx~\cite{CurticapeanDM17}, which is now often referred to as \emph{complexity monotonicity}. 
First, in Lemma~\ref{lem:bipcm}, we express $\#\cpbipindsubs{\Psi}{c}{\bG}$ as a linear combination of homomorphism counts. Then, in Lemma~\ref{lem:bip_mono}, we establish the corresponding ``complexity monotonicity'' --- intuitively this is a tight reduction from determining a homomorphism count that contributes to the linear combination to determining the entire linear combination, i.e., to $\#\cpbipindsubs{\Psi}{c}{\bG}$.

\begin{lemma}\label{lem:BipIndtoInd}
Let $\bG$ be a bipartite graph with a consistent $\bH$-colouring $c$, and let $\Psi$ be a bipartite property. Let $G$ and $H$ be the underlying graphs of $\bG$ and $\bH$, respectively. We have
\[\#\cpbipindsubs{\Psi}{c}{\bG} = \sum_{A\subseteq E(\bH)} \Psi(\bH[A]) \cdot \#\cpindsubs{H[A]}{c}{G}\,. \]
\end{lemma}
\begin{proof}
Recall from Section~\ref{sec:prelim} that, for each $A\subseteq E(H)$, $\cpindsubs{H[A]}{c}{G}$ is the set of all $S\subseteq V(G)$ for which $c\vert_S$ is an isomorphism from $G[S]$ to $H[A]$. Analogously, let $\cpbipindsubs{\bH[A]}{c}{\bG}$ be the set of all $S\subseteq V(G)$ for which $c\vert_S$ is a \emph{consistent} isomorphism from $\bG[S]$ to $\bH[A]$. Since $c$ is consistent by assumption of the lemma, we have 
\[
 \cpindsubs{H[A]}{c}{G}=\cpbipindsubs{\bH[A]}{c}{\bG}.
\]

Now recall that $\cpbipindsubs{\Psi}{c}{\bG}$ is the set of all subsets $S\subseteq V(G)$ of size $|V(\bH)|$ such that $c(S)=V(\bH)$ and $\Psi(\bG[S])=1$. We split this set according to the edges covered by $c(S)$: Let $S\in \cpbipindsubs{\Psi}{c}{\bG}$ and let $A_S=\{\{c(u),c(v)\}\in E(H) \mid \{u,v\}\in E(G[S])\}$. Since $c(S)=V(\bH)$ we observe that $c$ is a consistent homomorphism from $\bG[S]$ to $\bH[A_S]$, and that $c|_S$ is a consistent isomorphism from $\bG[S]$ to $\bH[A_S]$, i.e., that $S\in \cpbipindsubs{\bH[A_S]}{c}{\bG}$.
Summarising, we have
\begin{align*}
 \#\cpbipindsubs{\Psi}{c}{\bG} 
 &= \sum_{A\subseteq E(\bH)} \Psi(\bH[A]) \cdot \#\cpbipindsubs{\bH[A]}{c}{\bG}\\
 &= \sum_{A\subseteq E(\bH)} \Psi(\bH[A]) \cdot \#\cpindsubs{H[A]}{c}{G}\,.
\end{align*}
\end{proof}

\begin{lemma}\label{lem:bipcm}
Let $\Psi$ be a bipartite property and let $\bH$ be a bipartite graph. For every bipartite graph $\bG$ with consistent $\bH$-colouring $c$, we have
\begin{equation}\label{eq:bipcm}
    \#\cpbipindsubs{\Psi}{c}{\bG} = \sum_{T\subseteq E(H)} a_T \cdot \#\cphoms{H[T]}{c}{G}\,,
\end{equation}
    where $G$ and $H$ are the underlying graphs corresponding to $\bG$ and $\bH$, respectively, and $a_T=\sum_{A\subseteq T} \Psi(\bH[A])\cdot (-1)^{\#T-\#A}$.
\end{lemma}
\begin{proof}
In~\cite[Claim 1]{DorflerRSW22}, it was shown that
\[ \#\cpindsubs{H[A]}{c}{G} = \sum_{J \subseteq E(H)\setminus A} (-1)^{\#J} \cdot \#\cphoms{H[A\cup J]}{c}{G} \,.\]
In combination with Lemma~\ref{lem:BipIndtoInd}, we obtain:

\[
    \#\cpbipindsubs{\Psi}{c}{\bG} = \sum_{A\subseteq E(\bH)} \Psi(\bH[A]) \cdot \sum_{J \subseteq E(H)\setminus A} (-1)^{\#J} \cdot \#\cphoms{H[A\cup J]}{c}{G}.
\]
Observe that $A \cup J= A' \cup J'$ implies $\#\cphoms{H[A\cup J]}{c}{G}=\#\cphoms{H[A'\cup J']}{c}{G}$. Also note that $E(\bH)=E(H)$.
Thus, collecting the terms for each edge subset of $H$, we have
\[
    \#\cpbipindsubs{\Psi}{c}{\bG}=\sum_{T\subseteq E(H)} \Bigl( \sum_{A\subseteq T} \Psi(\bH[A])\cdot (-1)^{\#T-\#A} \Bigr) \cdot \#\cphoms{H[T]}{c}{G}.
\]
\end{proof}

\begin{lemma}[Bipartite Complexity Monotonicity]\label{lem:bip_mono} Let $\bH$ be a bipartite graph and let $\Psi$ be a computable bipartite property. There exists a computable function $f$ and an algorithm $\mathbb{A}$ that, given as input a graph $G$ with an $H$-colouring $c$, and given oracle access to the function
\[\#\cpbipindsubs{\Psi}{c}{\star} \,,\]
computes $\#\cphoms{H[T]}{c}{G}$ for each $T$ with $a_T\neq 0$, where the $a_T$ are the coefficients in~\eqref{eq:bipcm}. Furthermore, both the running time of $\mathbb{A}$  and the size $|V(\hat{\bG})|$ of every oracle query $\hat{\bG}$ are bounded by $f(|H|)\cdot |G|$.
\end{lemma}
\begin{proof}
Let $G$ be the given $H$-coloured graph with colouring $c$. Observe that the bipartition of $\bH$ together with the $H$-colouring $c$ induces a bipartition of $G$. Let $\bG$ be the corresponding bipartite graph.
Let $\id$ be the identity function on the vertices of $H$.
By Lemma~\ref{lem:bipcm} and Fact~\ref{fact:convenient}, for each subset $\hat{T}\subseteq E(H)$, we have
\[\#\cpbipindsubs{\Psi}{c}{(\bG, c)\timescc (\bH[\hat{T}], \id)} = \sum_{T\subseteq E(H)} a_T \cdot \#\cphoms{H[T]}{c}{(G,c) \timesc (H[\hat{T}],\id)}\,.\]
By~\cite[Lemma~5]{DorflerRSW22}, we have
\[\#\cphoms{H[T]}{c}{(G,c) \timesc (H[\hat{T}],\id)} = \#\cphoms{H[T]}{c}{G} \cdot \#\cphoms{H[T]}{c}{H[\hat{T}]}\,, \]
which allows us to rewrite
\[\#\cpbipindsubs{\Psi}{c}{(\bG,c)\timescc (\bH[\hat{T}],\id)} = \sum_{T\subseteq E(H)} b_T \cdot \#\cphoms{H[T]}{c}{H[\hat{T}]}\,,\]
where $b_T = a_T \cdot \#\cphoms{H[T]}{c}{G}$.
Calling the oracle for each $\hat{T}\subseteq E(H)$, this yields a system of linear equations, where the system matrix $M$ of size $2^{|E(H)|}\times 2^{|E(H)|}$ is given by
\[M(T,\hat{T}) = \#\cphoms{H[T]}{c}{H[\hat{T}]}\,.\]

It is known that $M$ is non-singular (even triangular)~\cite[Lemma 6]{DorflerRSW22}. Hence we can solve the system of linear equations and obtain as unique solution the values of $b_T$ for each $T\subseteq E(H)$. Consequently, for every $a_T\neq 0$, we can recover 
\[ \#\cphoms{H[T]}{c}{G} = b_T\cdot a_T^{-1}\,.\]

It remains to prove the bounds on the sizes of the oracle queries and on the running time. The former is immediate, since $|V((\bG,c) \timescc (\bH[\hat{T}],\id))|\leq |V(\bG)|\cdot |V(\bH)|$. For the latter, observe first, that the tensor product can also be constructed in time $|H|\cdot|G|$. Furthermore, the size of the matrix $M$, i.e., the number of equations, only depends on $H$, and not on $G$. In particular, by Gaussian elimination, we can certainly solve the system in time $\mathsf{poly}(2^{|H|})\cdot |V(G)|$. Finally, the computation of the coefficients $a_T=\sum_{A\subseteq T} \Psi(\bH[A])\cdot (-1)^{\#T-\#A}$ takes time only depending on $|H|$ and $\Psi$. Since $\Psi$ is fixed and computable\footnote{This is why we need computability of the properties for our hardness results. However, they would extend to non-computable properties under non-uniform FPT reductions, see~\cite{FlumG03}.}, the total time it takes to compute the $a_T$ is bounded by a computable function in $|H|$; this concludes the proof.
\end{proof}

The previous algorithm will serve as a reduction from  $\#\cphomsprob(\mathcal{H})$ to $\#\cpbipindsubsprob(\Psi)$, whenever $a_{E(H)}$ is non-zero for every $H\in \mathcal{H}$. 

The next lemma establishes a condition under which the coefficient that belongs to the $(k,k)$-biclique is non-zero.

\begin{lemma}\label{lem:bipartite_biclique_coeff}
Let $\bH=\bB_{k,k}$ for a prime $k$, and assume that $\Psi(\bB_{k,k})\neq\Psi(\bI_{k,k})$. Then the coefficient $a_{E(H)}$ in Equation~(\ref{eq:bipcm}) is non-zero.
\end{lemma}
\begin{proof}
Write $a\coloneqq a_{E(H)}$ and observe
\begin{equation}\label{eq:bcoef}
    a= \sum_{A\subseteq E(H)} \Psi(\bH[A])\cdot (-1)^{\#E(H)-\#A} = (-1)^{\#E(H)} \cdot \sum_{A\subseteq E(H)} \Psi(\bH[A])\cdot (-1)^{\#A}\,.
\end{equation}

First, for a bipartite graph $\bH=(V_1, V_2, E)$ with underlying graph $H$, we write $\auts{\bH}$ for the automorphisms of $H$ that are consistent automorphisms of $\bH$, i.e., that respect the bipartition $(V_1, V_2)$. Observe that $\auts{\bH}$ is a group under composition of functions, and that $\auts{\bH}$ acts transitively on the edges of $H$ (recall that $H=B_{k,k}$ is the biclique). Since $k$ is a prime, the number of edges of $H$ is a prime-power ($k^2$).  By the Orbit-stabiliser theorem and the fact that $\auts{\bH}$ acts transitively on $E(H)$ with $|E(H)|=k^2$, $|\auts{\bH}|$ is divisible by $k$. Thus, by Theorem~\ref{thm:Sylow}, a $k$-Sylow subgroup, say $\Gamma'$, must exist.
Furthermore, by Lemma~\ref{lem:Sylow}, $\Gamma'$ also acts transitively on $E(H)$.

Now, the action of $\Gamma'$ on $E(\bH)$ induces an action on the subsets $A$ of $E(\bH)$ by setting $g \triangleright A\coloneqq\allowbreak\{g \triangleright e \mid e\in A\}$. Since $\Gamma'$ only contains automorphisms that preserve the bipartition $(V_1,V_2)$, the subgraph $\bH[g\triangleright A]$ is well-defined, and $\bH[A]$ and $\bH[g\triangleright A]$ are consistently isomorphic for any $A\subseteq E(\bH)$ and $g\in \Gamma'$. In particular, the latter implies that $\Psi(\bH[A])=\Psi(\bH[g\triangleright A])$ for any $A\subseteq E(\bH)$ and $g\in \Gamma'$. Consequently, writing $\mathcal{O}$ for the set of all orbits of the action of $\Gamma'$ on subsets of $E(\bH)$, we can partition the sum in \eqref{eq:bcoef} as follows:
\[
    a= (-1)^{\#E(H)} \cdot \sum_{\orb_{\Gamma'}(A)\in \mathcal{O}} |\orb_{\Gamma'}(A)| \cdot \Psi(\bH[A])\cdot (-1)^{\#A}\,.
\]
Now recall that $\Gamma'$ is a $k$-Sylow subgroup and thus its order is a power of $k$. Since the size of each orbit must divide the order of $\Gamma'$, we obtain that $|\orb_{\Gamma'}(A)|=0\mod k$, unless $A$ is a fixed point. However, as $\Gamma'$ acts transitively on the edges, the only two fixed points of the induced action on edge-subsets are $A=\emptyset$ and $A=E(\bH)$. Finally, we observe that $\bH[\emptyset]$ is consistently isomorphic to $\bI_{k,k}$, and $\bH[E(\bH)]=\bH =\bB_{k,k}$. Hence 
\[a\equiv (-1)^{\#E(\bH)} \cdot \left(\Psi(\bI_{k,k}) + (-1)^{\#E(\bH)}\cdot \Psi(\bB_{k,k})\right) \equiv \Psi(\bB_{k,k}) + (-1)^{\#E(\bH)}\cdot \Psi(\bI_{k,k}) \mod k\,.\]
Since $\bH =\bB_{k,k}$, we have $\#E(\bH)=k^2$. Hence, if $k$ is a prime greater than $2$, we have 
\[a\equiv\Psi(\bB_{k,k}) - \Psi(\bI_{k,k}) \neq 0 \mod k \,,\]
as $\Psi(\bB_{k,k})\neq \Psi(\bI_{k,k})$ by assumption. For the same reason, if $k=2$, we have
\[a\equiv\Psi(\bB_{k,k}) + \Psi(\bI_{k,k}) \neq 0 \mod k \,.\]
Consequently, $a\neq 0$, which concludes the proof.
\end{proof}

Finally, we are able to put all of the previous pieces together and prove the main result of this section. To this end, recall that $\#\bipindsubsprob(\Psi)$ is the problem of, given as input a bipartite graph $\bG$ and a positive integer $k$, computing the number of all induced subgraphs of size $k$ of $\bG$ that satisfy $\Psi$. The parameterisation is given by $k$. Recall further that $\#\cpbipindsubsprob(\Psi)$ is the colour-prescribed version in which we expect as input a bipartite graph $\bG$ together with a consistent $\bH$-colouring~$c$, and the goal is to compute $\#\cpbipindsubs{\Psi}{c}{\bG}$; the parameter is $|V(\bH)|$.

A set of integers $\mathcal{K}$ is dense if there exists a constant $c$ such that for every positive integer $m$, there is a $k\in  \mathcal{K}$ with $m\leq k\leq c\cdot m$.
\mainbipartite*
\begin{proof}
Consider the following class of bicliques:
\[\mathcal{B}=\{B_{k,k} \mid k\text{ is prime and }\Psi(\bB_{k,k})\neq\Psi(\bI_{k,k})\}\,.\]
Recall that $\#\cphomsprob(\mathcal{B})$ is the problem that expects as input a $B_{k,k}$-coloured graph $G$ (for $B_{k,k}\in \mathcal{B}$), and the goal is to count the colour-prescribed homomorphisms from $B_{k,k}$ to $G$; the parameter is $|V(B_{k,k})|=2k$. It was shown in~\cite[Proof of Theorem 1]{DorflerRSW22} that for an infinite set of bicliques $\mathcal{B}$, the problem $\#\cphomsprob(\mathcal{B})$ is $\#\W{1}$-hard, and that for a dense set $\mathcal{B}$, it cannot be solved in time $f(k)\cdot |G|^{o(k)}$ for any function $f$, unless ETH fails. Now observe that from Lemmas~\ref{lem:bip_mono} and~\ref{lem:bipartite_biclique_coeff} it follows that $\#\cphomsprob(\mathcal{B})\tightred \#\cpbipindsubsprob(\Psi)$ (in particular, the parameter size does not increase, and the size of each oracle call is bounded by $f(|B_{k,k}|)\cdot |G|$). By Lemma~\ref{lem:tightred}, not only $\#\W{1}$-hardness, but also the tight lower bound under ETH transfers over to $\#\cpbipindsubsprob(\Psi)$. 

Finally, $\#\cpbipindsubsprob(\Psi) \tightred \#\bipindsubsprob(\Psi)$ via an easy application of the inclusion-exclusion principle: A proof is given in~\cite[Lemma 10]{DorflerRSW22} for graph properties, but it applies verbatim to bipartite properties, and the corresponding problems, as well. This concludes the proof.
\end{proof}

\section{From Bipartite Properties to Graph Properties}\label{sec:technical}

In this section, we prove Theorem~\ref{thm:technicalhardmain}, which states a hardness result for the problem $\#\indsubsprob(\Phi)$ subject to certain conditions of the graph property $\Phi$. These conditions are somewhat technical but they are also quite powerful. In particular, we will derive both our hardness result for hereditary properties (Section~\ref{sec:hereditary}) and the result for properties that are invariant under homomorphic equivalence (Section~\ref{sec:homequivalent}) from this result. We have not fully explored the implications of Theorem~\ref{thm:technicalhardmain}; we note, however, that it can also be applied to recover the known hardness result for the property of being ``connected''/ ``disconnected'' (one of the main results from~\cite{JerrumM15}) as well as establishing hardness results for other previously unclassified properties such as ``disconnected or bipartite''. Further examples are given in Section~\ref{sec:homequivalent}.

 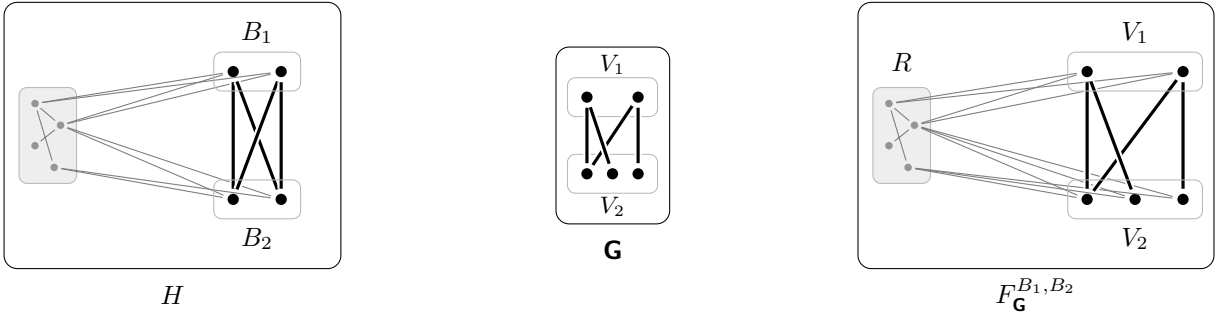
\begin{figure}[t]
        \centering
        \begin{tikzpicture}[scale=0.85]

            \begin{scope}[]
                \node (e1) at (-2.2, .8) {};
                \node (e2) at (.4, -2.8) {};
                \node (f1) at (2.5, .8) {};
                \node (f2) at (1.6, -2.8) {};
                \node[cvertex] (x1) at (-1.8, -1.5) {};
                \node[cvertex] (x2) at (-2, -.5) {};
                \node[cvertex] (x3) at (-1.6, -1) {};
                \node[fit=(x1)(x2)(x3), rounded corners=3pt, draw=lipicsGray!40,
                    fill=lipicsGray!10] {};
                \node[fit=(e1)(e2)(f1)(f2), rounded corners=5pt, draw] (f) {};
                \node at ($(f.south) + (0,-.4)$) {$H$};

                \node[cvertex] (x1) at (-1.7, -1.5) {};
                \node[cvertex] (x2) at (-2, -.5) {};
                \node[cvertex] (x3) at (-1.6, -.84) {};
                \node[cvertex] (x4) at (-2, -1.16) {};

                \node[vertex] (U4) at (1.1, -0) {};
                \node[vertex] (U2) at (1.85, -0) {};
                \node[vertex] (V3) at (1.85, -2) {};
                \node[vertex] (V4) at (1.1, -2) {};
                \draw[cedge] (U4) -- (x3) -- (V4);
                \draw[cedge] (x3) -- (V3);
                \draw[cedge] (U4) -- (x2);
                \draw[cedge] (V4) -- (x1) -- (x2) -- (x3) -- (x4);
                \draw[cedge] (V3) -- (x1);
                \draw[cedge] (U2) -- (x3);
                \draw[cedge] (U2) -- (x2);
                
                \draw[sdedge] (U4) -- (V4);
                \draw[sdedge] (U4) -- (V3);
                \draw[sdedge] (U2) -- (V4);
                \draw[sdedge] (U2) -- (V3);
                
                \node[fit=(U2)(U4), rounded corners=3pt, draw=lipicsGray!40] (x1) {};
                \node[fit=(V3)(V4), rounded corners=3pt, draw=lipicsGray!40] (x2) {};
                \node at ($(x1.north) + (0,.3)$) {$B_1$};
                \node at ($(x2.south) + (0,-.3)$) {$B_2$};
            \end{scope}

             \begin{scope}[xshift=13.75em]
                \node[vertex] (U3) at (2.6, -0.4) {};
                \node[vertex] (U4) at (1.8, -0.4) {};
                \node[vertex] (V2) at (2.6, -1.6) {};
                \node[vertex] (V3) at (2.2, -1.6) {};
                \node[vertex] (V4) at (1.8, -1.6) {};
                \node[fit=(U3)(U4), rounded corners=3pt, draw=lipicsGray!40] (x1) {};
                \node[fit=(V2)(V3)(V4), rounded corners=3pt, draw=lipicsGray!40] (x2) {};
                \node at ($(x2.south) + (0,-.2)$) {\small$V_2$};
                \node at ($(x1.north) + (0,.2)$) {\small$V_1$};
                \node (f1) at (2.81, .1) {};
                \node (f2) at (1.6, -2.1) {};
                \draw[sdedge] (U3) -- (V4);
                \draw[sdedge] (U3) -- (V2);
                \draw[sdedge] (V3) -- (U4) -- (V4);
                \node[fit=(f1)(f2), rounded corners=5pt, draw] (g) {};
                \node at ($(g.south) + (0,-.4)$) {$\bG$};
            \end{scope}

            \begin{scope}[xshift=38em]
             \node (e1) at (-2.2, .8) {};
                \node (e2) at (.4, -2.8) {};
                \node (f1) at (2.8, .8) {};
                \node (f2) at (1.6, -2.8) {};
                \node[cvertex] (x1) at (-1.8, -1.5) {};
                \node[cvertex] (x2) at (-2, -.5) {};
                \node[cvertex] (x3) at (-1.6, -1) {};
                \node[fit=(x1)(x2)(x3), rounded corners=3pt, draw=lipicsGray!40,
                    fill=lipicsGray!10] (z) {};
                \node at ($(z.north) + (0,.4)$) {$R$};
                \node[fit=(e1)(e2)(f1)(f2), rounded corners=5pt, draw] (f) {};
                \node at ($(f.south) + (0,-.4)$) {$F_{\bG}^{B_1, B_2}$};

                \node[cvertex] (x1) at (-1.7, -1.5) {};
                \node[cvertex] (x2) at (-2, -.5) {};
                \node[cvertex] (x3) at (-1.6, -.84) {};
                \node[cvertex] (x4) at (-2, -1.16) {};

                \node[vertex] (U3) at (2.6, -0) {};
                \node[vertex] (U4) at (1.1, -0) {};
                \node[vertex] (V2) at (2.6, -2) {};
                \node[vertex] (V3) at (1.85, -2) {};
                \node[vertex] (V4) at (1.1, -2) {};
                \draw[cedge] (U4) -- (x3) -- (V4);
                \draw[cedge] (U3) -- (x3) -- (V2);
                \draw[cedge] (x3) -- (V3);
                \draw[cedge] (U3) -- (x2);
                \draw[cedge] (U4) -- (x2);
                \draw[cedge] (V4) -- (x1) -- (x2) -- (x3) -- (x4);
                \draw[cedge] (V2) -- (x1);
                \draw[cedge] (V3) -- (x1);
                \draw[sdedge] (U3) -- (V4);
                \draw[sdedge] (U3) -- (V2);
                \draw[sdedge] (V3) -- (U4) -- (V4);
                \node[fit=(U3)(U4), rounded corners=3pt, draw=lipicsGray!40] (x1) {};
                \node[fit=(V2)(V3)(V4), rounded corners=3pt, draw=lipicsGray!40] (x2) {};
                \node at ($(x1.north) + (0,.3)$) {$V_1$};
                \node at ($(x2.south) + (0,-.3)$) {$V_2$};
         \end{scope}
        \end{tikzpicture}
        \caption{The implant of a bipartite graph $\bG$ into a pair of adjacent blocks $(B_1,B_2)$.}\label{fig:implant2}
    \end{figure}

\begin{defn}\label{def:FHBG}
Let $H$ be a graph, let $B_1$ and $B_2$ be non-empty disjoint sets of false twins\footnote{In our applications these sets of false twins will either be adjacent blocks (Section~\ref{sec:hereditary}) or the endpoints of an edge (Section~\ref{sec:homequivalent}).} (pairwise within the respective set) in $H$, and let $R= V(H)\setminus (B_1 \cup B_2)$.
Let $\bG=(V_1,V_2,E)$ be a bipartite graph. The \emph{\implant of $\bG$ into $(B_1,B_2)$}, is the following graph $F^{B_1, B_2}_\bG$ (see Figure~\ref{fig:implant2} for an illustration):
\begin{itemize}
    \item The vertices of $F^{B_1, B_2}_\bG$ are $V_1 \cup V_2 \cup R$.
    \item The edges of $F^{B_1, B_2}_\bG$ are partitioned into three groups:
    \begin{enumerate}
        \item[(a)] $E(F^{B_1, B_2}_\bG[R]) = E(H[R])$, that is, the edges between vertices in $R$ are as in $H$.
        \item[(b)] $E(F^{B_1, B_2}_\bG[V_1\cup V_2])= E(\bG)$, that is, the edges between $V_1$ and $V_2$ are as in $\bG$.
        \item[(c)] Finally, for every $v\in R$ with a neighbour in $B_1$\footnote{Since $B_1$ is a set of false twins, if $v$ is adjacent to one vertex in $B_1$, it is adjacent to all of them.}, we add all edges between $v$ and $V_1$. Similarly, if $v$ has a neighbour in $B_2$, we add all edges between $v$ and $V_2$.
    \end{enumerate}
\end{itemize}
Since the sets $B_1$ and $B_2$ (and implicitly the graph $H$) will always be clear from context, we drop the superscript from now on and simply write $F_{\bG}$.
Given a graph property $\Phi$, the \emph{\implant of $\Phi$ into $(B_1,B_2)$} is the bipartite property $\Psi^{B_1, B_2}_{\Phi}$ defined by $\Psi^{B_1, B_2}_{\Phi}(\bG)=\Phi(F_\bG)$.
Again, we drop the superscript from now on and simply write $\Psi_\Phi$.
Note that $\Psi_\Phi$ is well-defined since the graphs $F_\bG$ and $F_{\bG'}$ are, by construction, isomorphic whenever there is a \emph{consistent} isomorphism between $\bG$ and ${\bG'}$.
\end{defn}

We show a reduction from $\#\bipindsubsprob(\Psi_\Phi)$ to $\#\indsubsprob(\Phi)$ by constructing $F_\bG$ and using the inclusion-exclusion principle. A special case of this construction was performed in~\cite{RothSW20}.

\begin{lemma}\label{lem:implantreduction}
Let $\Phi$ be a computable graph property, let $H$ be a graph that contains non-empty sets of false twins $B_1$ and $B_2$, and let $\Psi_{\Phi}$ be the \implant of $\Phi$ into $(B_1, B_2)$.
Then $\#\bipindsubsprob(\Psi_\Phi) \tightred \#\indsubsprob(\Phi)$.
\end{lemma}
\begin{proof}
Let $(\bG, k)$ be an instance of $\#\bipindsubsprob(\Psi_\Phi)$ with $\bG=(V_1, V_2, E)$, and let $G$ be the underlying graph of $\bG$.
Let $R = V(H)\setminus(B_1 \cup B_2)$, and let $k' = |R| + k$. Further, let $F_{\bG}$ be the \implant of $\bG$ into $(B_1, B_2)$.
Note that the vertices of $F_{\bG}$ can be partitioned into the sets $V_1$, $V_2$ and $R$.
We show two claims about the set $\{S \in \indsubs{\Phi,k'}{F_{\bG}} \mid R \subseteq S\}$, i.e., the set of all induced subgraphs of size $k'$ in $F_{\bG}$ that satisfy $\Phi$ and that contain all vertices in $R$.

\begin{clm}\label{clm:equalCardinality}
We have
$
    \#\bipindsubs{\Psi_\Phi,k}{\bG} = \#\{S \in \indsubs{\Phi,k'}{F_{\bG}} \mid R \subseteq S\}.
$
\end{clm}
\begin{proofclaim}
\begin{align*}
    \#\{S \in \indsubs{\Phi,k'}{F_{\bG}} \mid R \subseteq S\}
    &= \#\{S\subseteq V(G) \mid \Phi(F_\bG[S\cup R])=1 \text{ and } |S|=k\}\\
    &= \#\{S\subseteq V(G) \mid \Phi(F_{\bG[S]})=1 \text{ and } |S|=k\}\\
    &= \#\{S\subseteq V(G) \mid \Psi_{\Phi}(\bG[S])=1 \text{ and } |S|=k\}\\
    &= \#\bipindsubs{\Psi_\Phi, k}{\bG}
\end{align*}
\end{proofclaim}

For each subset $J\subseteq R$, let $F_{\bG}\setminus J$ denote the graph $F_{\bG}[V(F_{\bG})\setminus J]$, i.e., the vertices in $J$ are deleted from $F_{\bG}$. We continue with a standard inclusion-exclusion argument.
\begin{clm}\label{clm:InclusionExclusion}
We have
$
\#\{S \in \indsubs{\Phi,k'}{F_{\bG}} \mid R \subseteq S\} = \sum_{J\subseteq R}{(-1)^{|J|} \cdot \indsubs{\Phi,k'}{F_{\bG}\setminus J}}.
$
\end{clm}
\begin{proofclaim}
Using the principle of inclusion and exclusion, we obtain that
\begin{align*} 
\#&\{S \in \indsubs{\Phi,k'}{F_{\bG}} \mid R \subseteq S\}\\
&= \#\indsubs{\Phi,k'}{F_{\bG}} - \#\{S \in \indsubs{\Phi,k'}{F_{\bG}} \mid   \exists\, v \in R\setminus S  \}\\
&=\#\indsubs{\Phi,k'}{F_{\bG}}  - \left| \bigcup_{v\in R} {\{ S \in \indsubs{\Phi,k'}{F_{\bG}} \mid v \notin S \}}  \right| \\ 
&=\#\indsubs{\Phi,k'}{F_{\bG}}  - \sum_{\emptyset \neq J \subseteq R}{ (-1)^{|J|+1} \left|  \bigcap_{v\in J}{ \{ S \in  \indsubs{\Phi,k'}{F_{\bG}} \mid v \notin S \} }  \right|  } \\ 
&=\#\indsubs{\Phi,k'}{F_{\bG}}  - \sum_{\emptyset \neq J \subseteq R}{ (-1)^{|J|+1} \#\{ S \in \indsubs{\Phi,k'}{F_{\bG}} \mid \forall\, v \in J: v \notin S \}  }\\ 
&=\#\indsubs{\Phi,k'}{F_{\bG}}  - \sum_{\emptyset \neq J \subseteq R}{ (-1)^{|J|+1} \#\indsubs{\Phi,k'}{F_{\bG}\setminus J}  }\\  
&=  \sum_{J \subseteq R}{ (-1)^{|J|}\cdot \#\indsubs{\Phi,k'}{F_{\bG}\setminus J}}. 
\end{align*}
\end{proofclaim}

By Claims~\ref{clm:equalCardinality} and~\ref{clm:InclusionExclusion}, an algorithm $\mathbb{A}$ that makes $2^{|R|}\in O(1)$ oracle calls, each of the form $(F_{\bG}\setminus J, k')$, can compute the sought-for value $\#\bipindsubs{\Psi_\Phi,k}{\bG}$ in time linear in $|G|$. In particular, $|V(F_{\bG})\setminus J|\in O(|V(G)|)$ for each $J\subseteq R$ and $k'\in O(k)$. This completes the proof.
\end{proof}

The following result is a direct consequence of Theorem~\ref{thm:main_bipartite} and Lemma~\ref{lem:implantreduction}.
\begin{theorem}\label{thm:technicalhardmain}
Let $\Phi$ be a computable graph property, let $H$ be a graph that contains sets of false twins $B_1$ and $B_2$, and let $\Psi_{\Phi}$ be the \implant of $\Phi$ into $(B_1, B_2)$.
Let $\mathcal{K}$ be the set of primes $k$ for which $\Psi_{\Phi}(\bI_{k,k})\neq \Psi_{\Phi}(\bB_{k,k})$.
If $\mathcal{K}$ is infinite then $\#\indsubsprob(\Phi)$ is $\#\W{1}$-hard.
Moreover, if $\mathcal{K}$ is dense then $\#\indsubsprob(\Phi)$ cannot be solved in time $f(k)\cdot |G|^{o(k)}$ for any function $f$, assuming ETH.
\end{theorem}

\section{Hereditary Graph Properties}\label{sec:hereditary}

Let $\Phi$ be a hereditary graph property. Recall from Observation~\ref{obs:hereditarynontriv} that $\Phi$ is not \meagre if $\Phi$ is true for infinitely many graphs and there is at least one forbidden induced subgraph, i.e., $\obstr(\Phi)\neq \emptyset$. Recall further that $\overline{\Phi}$ is the inverse of $\Phi$. A graph $H'$ is an induced subgraph of a graph $H$ if and only if $\overline{H'}$ is an induced subgraph of $\overline{H}$. Consequently, $\overline{\Phi}$ is hereditary if and only if $\Phi$ is hereditary, and for the set of forbidden induced subgraphs of $\overline{\Phi}$ we have $\obstr(\overline{\Phi})=\{\overline{H} \mid H\in \obstr(\Phi) \}$.

What follows is a simple application of Ramsey's theorem. Suppose that $\obstr(\Phi)$ contains an independent set of size $c_1$, and that $\obstr(\overline{\Phi})$ contains an independent set of size $c_2$. Then $\obstr(\Phi)$ contains a clique of size $c_2$. Let $F$ be a graph with $\Phi(F)=1$. Then $F$ does not contain an independent set of size $c_1$, and it also does not contain a clique of size $c_2$ as an induced subgraph. By Ramsey's theorem, $|V(F)|$ is bounded by the Ramsey number $R(c_1,c_2)$. As a consequence, $\Phi$ must be false on all graphs with more than $R(c_1,c_2)$ vertices; hence $\Phi$ is \meagre. We formally state this observation.

\begin{obs}\label{obs:ramsey}
Let $\Phi$ be a hereditary graph property. If $\obstr(\Phi)$ and $\obstr(\overline{\Phi})$ contain an independent set, then $\Phi$ is \meagre.
\end{obs}

Recall that a \emph{block} in a graph $H$ is an equivalence class of the vertices with respect to the ``false twin relation''.
\begin{lemma}\label{lem:psi_phi_valid}
Let $\Phi$ be a hereditary graph property that is not \meagre. Suppose that every graph in $\obstr(\Phi)$ has at least one edge. Then there is a graph $H$ with blocks $B_1$ and $B_2$ such that for the \implant $\Psi_\Phi$ of $\Phi$ into $(B_1, B_2)$, it holds that, for each $k>|V(H)|$, we have $\Psi_\Phi(\bB_{k,k})\neq \Psi_\Phi(\bI_{k,k})$.
\end{lemma}
\begin{proof}
Let $H\in \obstr(\Phi)$ be a forbidden induced subgraph of $\Phi$ such that $H\down$ has a minimum number of edges among all graphs in
    $\{H\down \mid H\in \obstr(\Phi)\}$.
Since $H$ has at least one edge but no self-loops, it contains at least two blocks $B_1$ and $B_2$ that are connected by an edge in $H\down$. Consequently, $H[B_1 \cup B_2]$ must be a complete bipartite graph. 
For a bipartite graph $\bG$, let $F_\bG$ be the \implant of $\bG$ into $(B_1, B_2)$. Further, let $\Psi_\Phi$ be the \implant of $\Phi$ into $(B_1, B_2)$.
For this choice we will show that, for each $k>|V(H)|$, we have $\Psi_\Phi(\bB_{k,k})\neq \Psi_\Phi(\bI_{k,k})$.

Let $k>|V(H)|$.
We have $\Psi_\Phi(\bB_{k,k})=0$ since, for $k>|V(H)|$, $F_{\bB_{k,k}}$ contains $H$ as induced subgraph, and thus $\Phi(F_{\bB_{k,k}})=0$.

Next we show that $\Psi_\Phi(\bI_{k,k})=1$, i.e., that $\Phi(F_{\bI_{k,k}})=1$. In other words, we have to show that $F_{\bI_{k,k}}$ does not contain any forbidden induced subgraph from $\obstr(\Phi)$. Suppose $F'$ is an induced subgraph of $F_{\bI_{k,k}}$. By Lemma~\ref{lem:X}, we have $|E(F'\down)|\leq |E(F_{\bI_{k,k}}\down)|$. 
Next, we show that $F_{\bI_{k,k}}\down$ has strictly fewer edges than $H\down$: Observe that, by definition of the construction $\bG \mapsto F_\bG$, the graph $F_{\bI_{k,k}}$ is obtained from $H$ by first deleting all edges between the blocks $B_1$ and $B_2$ and then adding false twins to $B_1$ and $B_2$ until both blocks contain $k$ vertices. We distinguish two cases:
\begin{enumerate}
    \item $B_1\cup B_2$ is a single block of false twins in $F_{\bI_{k,k}}$ --- observe that this can happen if, in $H$, the vertices in $B_1\cup B_2$ have the same neighbours outside of $B_1\cup B_2$, that is, there exists $R\subseteq V(H)\setminus(B_1\cup B_2)$ such that
    \[N_H(B_1)\setminus B_2 = R = N_H(B_2)\setminus B_1 \,.\]
    In that case, the graph $F_{\bI_{k,k}}\down$ is obtained from $H\down$ by contracting the edge $\{B_1,B_2\}$ in $H\down$ and deleting self-loops and multi-edges afterwards. Consequently, $F_{\bI_{k,k}}\down$ has (at least) one edge less than $H\down$.
    \item Otherwise, $B_1$ and $B_2$ remain distinct equivalence classes. Consequently, $F_{\bI_{k,k}}\down$ is obtained by deleting the edge $\{B_1,B_2\}$ in $H\down$. We conclude that, again, $F_{\bI_{k,k}}\down$ has one edge less than $H\down$.
\end{enumerate}
Summarising, we have
\[|E(F'\down)|\leq |E(F_{\bI_{k,k}}\down)| < |E(H\down)|\,. \]
Since $H$ was chosen such that $|E(H\down)|$ is minimal among all graphs in $\obstr(\Phi)$, we conclude that $F'$ is not contained in $\obstr(\Phi)$. Consequently, no induced subgraph of $F_{\bI_{k,k}}$ is contained in $\obstr(\Phi)$, and therefore $\Phi(F_{\bI_{k,k}})=1$.
\end{proof}

We can now show the main result of this work, which also confirms Conjecture~\ref{conj:main} for all hereditary properties.

\hereditary*
\begin{proof}
If $\Phi$ is \meagre, we have that for all but finitely many $k$, $\Phi$ is either trivially true, or trivially false on $k$-vertex graphs. Since $\Phi$ is hereditary, that is, closed under vertex deletion, the latter implies that $\Phi$ is either constant true, or there exists a constant $B$ such that $\Phi$ is false on all graphs with at least $B$ vertices. In this special case, the algorithm in Lemma~\ref{lem:meagre_easy} runs in polynomial-time.

Now assume $\Phi$ is a hereditary graph property that is not \meagre. First suppose that no forbidden induced subgraph of $\Phi$ is an independent set. Then, according to Lemma~\ref{lem:psi_phi_valid}, there is an \implant $\Psi_\Phi$ such that, for each $k > |V(H)|$, we have $\Psi_\Phi(\bB_{k,k})\neq \Psi_\Phi(\bI_{k,k})$. In particular, this holds for all but a finite number of primes $k$. The statement of the theorem then follows from Theorem~\ref{thm:technicalhardmain}.

Now suppose that $\Phi$ has an independent set as forbidden induced subgraph. Since $\Phi$ is not \meagre, by Lemma~\ref{obs:ramsey}, its inverse $\overline{\Phi}$ does not contain an independent set as forbidden induced subgraph. Consequently, $\overline{\Phi}$ is true for all independent sets and, as $\Phi$ has at least one forbidden induced subgraph, $\overline{\Phi}$ does as well. So, by Observation~\ref{obs:hereditarynontriv}, $\overline{\Phi}$ is not \meagre. Recall that $\overline{\Phi}$ is hereditary if and only if $\Phi$ is hereditary. Now using the same argument as before, by Lemma~\ref{lem:psi_phi_valid} and Theorem~\ref{thm:technicalhardmain}, $\#\indsubsprob(\overline{\Phi})$ is $\#\W{1}$-hard and, assuming ETH, cannot be solved in time $f(k)\cdot |G|^{o(k)}$ for any function $f$.
Finally, it is not hard to see that $\#\indsubsprob(\overline{\Phi})\tightred \#\indsubsprob(\Phi)$~\cite[Fact 2.3]{RothSW20} (in fact, the two problems are tightly interreducible).
\end{proof}

\section{Graph Properties Invariant under Homomorphic Equivalence}\label{sec:homequivalent}

A graph property $\Phi$ is \emph{closed under homomorphisms} if, for each pair of graphs $H_1$ and $H_2$ with a homomorphism from $H_1$ to $H_2$, $\Phi(H_1)=1$ implies $\Phi(H_2)=1$. As an example, being ``non-bipartite'' is closed under homomorphisms. Conjecture~\ref{conj:main} is known to hold for all properties that are closed under homomorphisms as a consequence of~\cite[Theorem 2]{DorflerRSW22}. Their result can be applied since properties that are closed under homomorphisms are also supergraph-closed, i.e., if $H_1$ is a subgraph of $H_2$ then $\Phi(H_1)=1$ implies $\Phi(H_2)=1$. (Using the terminology from~\cite{DorflerRSW22}, this means that $\neg\Phi$ is monotone.)

In this section, we consider a more general class of properties. Two graphs $H_1$ and $H_2$ are \emph{homomorphically equivalent} if there is a homomorphism from $H_1$ to $H_2$ and vice versa.
Recall that a graph property $\Phi$ is \emph{invariant under homomorphic equivalence} if, for each pair of homomorphically equivalent graphs $H_1$ and $H_2$, we have $\Phi(H_1)=\Phi(H_2)$.
Slightly counterintuitively, this leads to a more general class of properties:
Note that if $\Phi$ is closed under homomorphisms then it is also invariant under homomorphic equivalence, but not necessarily the other way around. Thus, confirming Conjecture~\ref{conj:main} for the class of properties that are invariant under homomorphic equivalence is a strictly more general result than a confirmation for all properties that are closed under homomorphisms.

Examples for properties that are invariant under homomorphic equivalence (but not closed under homomorphisms) are ``odd girth equal to $c$'', ``chromatic number equal to $c$'', ``clique number equal to $c$'' etc. Such properties are not subgraph-closed/supergraph-closed, hereditary or edge-monotone. 
In Theorem~\ref{thm:homequivalence}, we confirm Conjecture~\ref{conj:main} for all properties that are invariant under homomorphic equivalence. This result is actually a corollary of Theorem~\ref{thm:twinequivalence}, which covers a less natural but even more general class of properties.

Two graphs $H_1$ and $H_2$ are \emph{twin-equivalent} if the corresponding twin-free quotients $H_1\down$ and $H_2\down$ are isomorphic.
A graph property $\Phi$ is \emph{twin-invariant} if, for each pair of twin-equivalent graphs $H_1$ and $H_2$ with at least two vertices, $\Phi(H_1)=\Phi(H_2)$.

\begin{lemma}\label{lem:psi_phi_twin}
Let $\Phi$ be a twin-invariant graph property that is not \meagre. Then there is a graph $H$ with an edge $\{u,v\}$ such that for the \implant $\Psi_\Phi$ of $\Phi$ into $(\{u\}, \{v\})$, it holds that, for every integer $k\ge 1$, we have $\Psi_\Phi(\bB_{k,k})\neq \Psi_\Phi(\bI_{k,k})$.
\end{lemma}
\begin{proof}
Since $\Phi$ is not \meagre, for some $n\ge 2$, there are $n$-vertex graphs $H_1$ and $H_2$ for which $\Phi(H_1)\neq \Phi(H_2)$. Consequently, there is some  $n$-vertex graph $H$ with an edge $e=\{u,v\}$ such that $\Phi(H)\neq \Phi(H-e)$, where $H-e$ denotes the graph with vertices $V(H)$ and edges $E(H)\setminus \{e\}$. Let $\Psi_\Phi$ be the \implant of $\Phi$ into $(\{u\}, \{v\})$. Accordingly, for a bipartite graph $\bG$, let $F_{\bG}$ be the \implant of $\bG$ into $(\{u\}, \{v\})$.

Note that, for each $k\ge 1$, $F_{\bB_{k,k}}$ and $H$ are twin-equivalent and have at least $2$ vertices. Similarly, $F_{\bI_{k,k}}$  and $H-e$ are twin-equivalent and have at least two vertices. Thus, since $\Phi$ is twin-invariant, 
\[
    \Psi_\Phi(\bB_{k,k})= \Phi(F_{\bB_{k,k}}) = \Phi(H) \neq \Phi(H-e) = \Phi(F_{\bI_{k,k}}) = \Psi_\Phi(\bI_{k,k}).
\]
\end{proof}

The following theorem is an immediate consequence of Lemma~\ref{lem:psi_phi_twin} and Theorem~\ref{thm:technicalhardmain}.

\twinequivalence*
First assume that $\Phi$ is \meagre. Since $\Phi$ is twin-invariant, it is, in particular, closed under the addition of false twins. It is easy to see that the latter implies $\Phi$ being either constant false, or that there exists a constant $B$ such that $\Phi$ is true for all graphs with at least $B$ vertices. In this special case, the algorithm of Lemma~\ref{lem:meagre_easy} runs in polynomial time.

Now assume that $\Phi$ is not \meagre.
Note that the twin-free quotient of a connected graph is also connected. If $H$ is disconnected and has at least one edge then its twin-free quotient $H\down$ is also disconnected. Therefore, a connected graph $H_1$ can only be twin-equivalent to a disconnected graph if $H_1$ is the graph with a single vertex and no edge. Therefore, if $H_1$ and $H_2$ are twin-equivalent graphs with at least two vertices, $H_1$ is disconnected if and only if $H_2$ is disconnected. This shows that the property of being ``disconnected'' is twin-invariant, and Theorem~\ref{thm:twinequivalence} can be applied in this case.  
Theorem~\ref{thm:twinequivalence} also covers similar properties, for which the complexity of $\#\indsubsprob$ was previously unclassified, such as ``disconnected or bipartite'' or ``disconnected or triangle-free''.\footnote{By allowing all bipartite graphs,~\cite[Theorem 1]{DorflerRSW22} cannot be applied.}

More importantly, as $H$ and $H\down$ are homomorphically equivalent, every graph property that is invariant under homomorphic equivalence is also twin-invariant. Hence, the following statement is a special case of Theorem~\ref{thm:twinequivalence}.

\homclosed*

\bibliographystyle{plainurl}
\bibliography{IndSub}

\end{document}